%% file: main.tex
\documentclass[11pt]{article}
\input{preamble}

\title{Online Edge Coloring is (Nearly) as Easy as Offline}

\author[1]{Joakim Blikstad\thanks{Supported by the Swedish Research Council (Reg. No. 2019-05622) and the Google PhD Fellowship Program. Work done in part while visiting EPFL.}}
\author[2]{Ola Svensson\thanks{Supported by the Swiss State Secretariat for Education, Research and Innovation (SERI) under contract number MB22.00054.}}
\author[2]{Radu Vintan\protect\footnotemark[\value{footnote}]}
\author[3]{David Wajc\thanks{Supported by a Taub Family Foundation ``Leader in Science and Technology'' fellowship. Work done in part while the author was at Google Research.}}

\affil[1]{KTH Royal Institute of Technology \& Max Planck Institute for Informatics}
\affil[2]{EPFL}
\affil[3]{Technion --- Israel Institute of Technology}

\date{\vspace{-1.3cm}}

\begin{document}
\maketitle

\pagenumbering{gobble}
\input{Sections/abstract}

\newpage

\tableofcontents
\newpage
\pagenumbering{arabic}
\input{Sections/intro}
\input{Sections/prelims}
\input{Sections/intuition}
\input{Sections/analysis}

\appendix
\section*{APPENDIX}
\input{Sections/applications}
\input{Sections/list-coloring}
\input{Sections/local-result}
\input{Sections/extension}
\input{Sections/unknown}

\paragraph{Acknowledgements.} David Wajc thanks Nati Linial for asking him about online list edge coloring.

\bibliographystyle{alpha}
\bibliography{abb,ultimate,bibliography}

\end{document}

%% file: preamble.tex
\usepackage[english]{babel}
\usepackage{authblk}
\usepackage{inconsolata}
\usepackage{libertine}
\usepackage{amsmath,amssymb,amsthm,mathrsfs}
\usepackage{mathtools}
\usepackage{mathrsfs,bbold}
\usepackage{xcolor}
\definecolor{ForestGreen}{rgb}{0.1333,0.5451,0.1333}
\definecolor{DarkRed}{rgb}{0.65,0,0}
\definecolor{Red}{rgb}{1,0,0}
\usepackage[linktocpage=true,backref=page,
pagebackref=true,colorlinks,
linkcolor=DarkRed,citecolor=ForestGreen,
bookmarks,bookmarksopen,bookmarksnumbered]
{hyperref}

\usepackage{comment}
\usepackage{mdframed}
\usepackage{xspace}

\usepackage{url}
\usepackage[noend]{algorithmic}

\usepackage[colorinlistoftodos]{todonotes}

\allowdisplaybreaks

\renewcommand{\leq}{\leqslant}
\renewcommand{\geq}{\geqslant}
\renewcommand{\le}{\leqslant}
\renewcommand{\ge}{\geqslant}

\newcommand{\calM}{\mathcal{M}}

\newcommand{\E}{\mathbb{E}}

\DeclarePairedDelimiter\ceil{\lceil}{\rceil}

\usepackage[margin=1in]{geometry}
\usepackage{caption}

\newcommand{\eps}{\varepsilon}

\usepackage{todonotes}

\newcommand{\calA}{\mathcal{A}}

\newcommand{\edgearrivalalg}{\textsc{MatchingAlgorithm}\xspace}
\newcommand{\edgearrivalalgext}{\textsc{MatchingAlgorithmChanges}\xspace}

\usepackage[capitalize]{cleveref}
\usepackage{dsfont}

\usepackage{thmtools}
\declaretheorem[numberwithin=section,refname={Theorem,Theorems},Refname={Theorem,Theorems}]{theorem}

\declaretheorem[numberlike=theorem]{lemma}

\declaretheorem[numberlike=theorem,style=definition]{definition}
\declaretheorem[numberlike=theorem]{claim}
\declaretheorem[numberlike=theorem]{conjecture}
\declaretheorem[numberlike=theorem]{remark}
\declaretheorem[numberlike=theorem,refname={Fact,Facts},Refname={Fact,Facts},name={Fact}]{fact}

\declaretheorem[numberlike=theorem, refname={Observation,Observations},Refname={Observation,Observations},name={Observation}]{observation}

\declaretheorem[refname={Algorithm,Algorithms},Refname={Algorithm,Algorithms},name={Algorithm}]{algorithm}

%% file: Sections/abstract.tex
\begin{abstract}
   The classic theorem of Vizing (Diskret.~Analiz.'64) asserts that any graph of maximum degree $\Delta$ can be edge colored (offline) using no more than $\Delta+1$ colors (with $\Delta$ being a trivial lower bound).
   In the \emph{online} setting, 
   Bar-Noy, Motwani and Naor (IPL'92) conjectured that a $(1+o(1))\Delta$-edge-coloring can be computed online in $n$-vertex graphs of maximum degree $\Delta=\omega(\log n)$.
    Numerous algorithms made progress on this question, using a higher number of colors or assuming restricted arrival models, such as random-order edge arrivals or vertex arrivals (e.g., AGKM~FOCS'03, BMM~SODA'10, CPW~FOCS'19, BGW~SODA'21, KLSST~STOC'22).
    In this work, we resolve this longstanding conjecture in the affirmative in the most general setting of adversarial edge arrivals.
    We further generalize this result to obtain online counterparts of the \emph{list} edge coloring result of Kahn (J. Comb. Theory. A'96) and of the recent ``local'' edge coloring result of Christiansen (STOC'23).
\end{abstract}

%% file: Sections/intro.tex
\section{Introduction}

Online edge coloring is one of the first problems studied through the lens of competitive analysis \cite{bar1992greedy}. 
In this problem, a graph is revealed piece by piece (either edge-by-edge or vertex-by-vertex). 
An algorithm must assign colors to edges upon their arrival irrevocably 
so that no two adjacent edges are assigned the same color.
The algorithm's objective is to use few colors in any graph of maximum degree $\Delta$,   close to the (offline) optimal $\Delta$ or $\Delta+1$ guaranteed by Vizing's Theorem \cite{vizing1964estimate}.

A trivial greedy online algorithm that assigns \emph{any} available color
to each edge upon arrival succeeds while using a palette of $2\Delta-1$ colors.
As shown over three decades ago, no algorithm does better on {low-degree} $n$-vertex graphs with sufficiently small maximum degree $\Delta=O(\log n)$ \cite{bar1992greedy}. While our understanding of online edge coloring is thus complete   on \emph{low-degree graphs}, the dynamics are considerably more complex and interesting when 
$\Delta$ surpasses $\omega(\log n)$.  The authors in~\cite{bar1992greedy} even conjectured that barring the low-degree case, 
edge coloring can be performed online while nearly matching the guarantees of offline methods.

\begin{conjecture}[\cite{bar1992greedy}]\label{conj:BMN}
There exists an online edge-coloring algorithm for $n$-vertex graphs that colors the edges of the graph online using $(1 + o(1))\Delta$ colors, assuming known maximum degree $\Delta = \omega(\log n)$.\footnote{Knowledge of $\Delta$ is necessary: no algorithm can use fewer than $\frac{e}{e-1}\Delta\approx 1.582\Delta$ colors otherwise \cite{cohen2019tight}.}
\end{conjecture}

Progress towards resolving \Cref{conj:BMN} was obtained for restricted settings, including  
random-order edge arrivals \cite{aggarwal2003switch,bahmani2012online,bhattacharya2021online,kulkarni2022online} and vertex arrivals \cite{cohen2019tight,saberi2021greedy,blikstad2023simple}. 
In the most general setting, i.e., under adversarial edge arrivals, \cite{kulkarni2022online} recently provided the first algorithm outperforming the trivial algorithm, showing that an $(\frac{e}{e-1}+o(1))\Delta$-edge-coloring is achievable. 

Most prior results for online edge coloring under adversarial arrivals \cite{cohen2019tight,saberi2021greedy,kulkarni2022online} were attained via the following tight connection between online edge coloring and online matching.\footnote{Online matching algorithms must decide for each arriving edge whether to irrevocably add it to their output matching.}
Given an $\alpha\Delta$-edge-coloring algorithm, it is easy, by sampling a color, to obtain an online matching algorithm that matches each edge with probability at least $1/(\alpha\Delta)$. In contrast, \cite{cohen2019tight} provided an (asymptotically) optimal reduction in the opposite direction, from online $(\alpha+o(1))\Delta$-edge-coloring algorithms when $\Delta=\omega(\log n)$ to online matching algorithms that match each edge with probability at least $1/(\alpha\Delta)$. (See \Cref{lemma:reduction2}.) A positive resolution of \Cref{conj:BMN} therefore requires---and indeed, is equivalent to---designing an online matching algorithm that matches each edge with probability at least $1/((1+o(1))\Delta)$.

\paragraph{Our Main Result.} In this work, we resolve  \Cref{conj:BMN}:

\begin{theorem} [See exact bounds in \Cref{thm:ec-body}]\label{thm:edge-coloring}
    There exists an online algorithm that, on $n$-vertex graphs with known maximum degree $\Delta=\omega(\log n)$, outputs a $(1+o(1))\Delta$-edge-coloring with high probability.
\end{theorem}

Via the aforementioned reduction, we obtain the above from our following key technical contribution.

\begin{restatable}{theorem}{OnlineMatchingTheorem} \label{thm:online_matching}
    There exists an online matching algorithm that on graphs with known maximum degree $\Delta$, outputs a random matching $M$ satisfying
    $$\Pr[e\in M] \ge \frac{1}{\Delta + \Theta(\Delta^{3/4} \log^{1/2}\Delta)} = \frac{1}{(1+o(1)) \cdot \Delta} \qquad \forall e\in E.$$
\end{restatable}

We note that the above matching probability of $1/(\Delta+q)$ for $q =  \Theta(\Delta^{3/4} \log^{1/2}\Delta)$ approaches a (lower order) lower bound of $q = \Omega(\sqrt{\Delta})$ implied by the competitiveness lower bound of $1-\Omega(1/\sqrt{\Delta})$ for online matching in regular graphs due to \cite{cohen2018randomized}.

Before explaining further implications of our results and techniques, we briefly discuss our approach for \Cref{thm:online_matching} and its main differences compared to prior work. 
 
\paragraph{Techniques overview.} For intuition, consider a simple ``algorithm'' that, by its very design, appears to match every edge with a probability of $1/(\Delta + q)$: 
\begin{center}
\begin{minipage}{0.95\textwidth}
\begin{mdframed}[hidealllines=true, backgroundcolor=gray!15]
When an edge $e_t = (u,v)$ arrives and connects two unmatched vertices, match it with probability  \[P(e_t) = \frac{1}{\Delta + q}\cdot \frac{1} {\Pr[u,v \textrm{ both unmatched until time $t$}]}\,.\]
\end{mdframed}
\end{minipage}
\end{center}
However, the caveat, and the reason for the quotation marks around "algorithm," is that this process is viable only if $P(e_t)$ constitutes a probability. 
 This raises the question: how large must $q$ 
 be to ensure that $P(e_t) \leq 1$ for all edges $e_t$? Suppose we naively assume that the events ``$u$ is unmatched until time $t$'' and ``$v$ is unmatched until time $t$'' are independent. In that case, straightforward calculations show that $q = O(\sqrt{\Delta})$ would suffice for well-defined probabilities.
  However, the assumption of independence rarely holds outside of simplistic graphs like trees, and so the aforementioned events may exhibit complex and problematic correlations. 
  Such correlations present the central challenge in establishing tight bounds for general edge arrivals. 
  
  Previous studies addressed the above challenge by circumscribing and managing these correlations.
  For instance, in more constrained arrival models, \cite{cohen2018randomized,cohen2019tight} used a variant of this approach; they rely heavily on one-sided vertex arrivals in bipartite graphs to choose an edge to match in a correlated way upon each vertex arrival, while creating useful negative correlation allowing for Chernoff bounds beneficial for future matching choices.
  In contrast, the only known method for general edge arrivals was given by \cite{kulkarni2022online}: they subsample locally tree-like graphs and employ sophisticated correlation decay techniques to approximate the independent scenario, albeit at the expense of only being able to match each edge with a probability of at least $1/(\alpha\Delta)$ for $\alpha:=e/(e-1)+o(1)$.
Unfortunately, the ratio of $e/(e-1)$ appears to be an intrinsic barrier for this approach.

Our approach deviates from the one guiding~\cite{cohen2018randomized,cohen2019tight,kulkarni2022online}, by allowing for correlations instead of controlling and taming them.
Crucially, we present a different but still simple algorithm, with a subtle difference,  which we describe informally here (see detailed exposition in \Cref{sec:intuition}): instead of obtaining the probability $P(e_t)$ by scaling $1/(\Delta+q)$ by $1/\Pr[u,v \textrm{ both unmatched until time $t$}]$, our scaling factor depends upon the algorithm's actual execution path (sequence of random decisions) so far.
The modified algorithm allows us to analyze the scaling factor for an edge as a martingale process. 
While there may still be correlations,  we show that this martingale has (i) small step size and (ii) bounded observed variance. 
These properties allow for strong Chernoff-type concentration bounds, specifically through Freedman's inequality (\Cref{thm:freedman_inequality}), which is pivotal to our analysis. 
The change of viewpoint is crucial for achieving our result and leads to a simple and concise algorithm and analysis. 
The results and techniques also extend to more general settings, as we explain next.

\paragraph{Secondary Results and Extensions.}

In \Cref{sec:matching_peeling}  we combine \Cref{thm:online_matching} with a new extension of the above-mentioned reduction, from which we obtain online counterparts to two (offline) generalizations of Vizing's theorem, concerning both ``local'' and list edge coloring.
For some background, a \emph{list edge coloring} of a graph is a proper coloring of the edges, assigning each edge a color from an edge-specific palette.
The \emph{list chromatic number} of a graph, also introduced by Vizing \cite{vizing1976vertex}, is the least number of colors needed for each edge to guarantee that a proper list edge coloring exists.
A seminal result of Kahn \cite{kahn1996asymptotically} shows that the list chromatic number is asymptotically equal to $\Delta$.
Another, ``local'', generalization of Vizing's Theorem was recently obtained by Christiansen \cite{christiansen2023power}, who showed that any graph's edges can be properly colored (offline) with each edge $(u,v)$ assigned a color in the set $\{1,2,\dots,1+\max(\deg(u),\deg(v))\}$.
In this work, using \Cref{thm:online_matching} and extensions of the aforementioned reduction, we show that results of the same flavor as \cite{kahn1996asymptotically} and \cite{christiansen2023power} can be obtained by \emph{online algorithms}.

\begin{theorem}
[See exact bounds in \Cref{thm:list_edge_coloring}] There exists an online algorithm that computes an edge coloring which, with high probability, assigns each edge $e$ a color from its list $L(e)$ (revealed online, with edge~$e$), provided each list has sufficiently large size $(1+o(1))\Delta$ and that $\Delta=\omega(\log n)$.
\end{theorem}

\begin{theorem}
[See exact bounds in \Cref{thm:local_edge_coloring_strenghtened}]
There exists an online algorithm that computes an edge coloring assigning each edge $e=(u,v)$ a color from the set $\left\{ 1,\,2\,,\,\ldots,\, d_{\text{max}}(e)\cdot (1+o(1)) \right \}$ with high probability, where $d_{\text{max}}(e) := \max\{\deg(u), \deg(v)\}$, 
provided that $d_{\text{max}}(e) = \omega(\log n)$.
\end{theorem}

Finally, in \Cref{sec:extension}, we show that our algorithmic and analytic approach underlying \Cref{thm:online_matching} allows us more generally to round fractional matchings online.
Here, an \emph{$\alpha$-approximate online rounding algorithm} for fractional matchings is revealed (online) an assignment of non-negative values to the edges, $x:E\to \mathbb{R}_{\geq 0}$, with value $x_{e}$ revealed upon arrival of edge $e$, so that the total assigned value to the incident edges to each vertex is at most one, $\sum_{e\ni v}x_e\leq 1$.  
The algorithm's objective is to output online a randomized matching $M$ that matches each edge $e$ with probability at least $x_e/\alpha$. For \emph{one-sided vertex arrivals} in \emph{bipartite} graphs,  it is known that the optimal $\alpha$ is in the range $(1.207,1.534)$ \cite{naor2023online}, while if the matching is ``sufficiently spread out'', $\max_e x_e\leq o(1)$, then $\alpha=1+o(1)$ is possible \cite[Chapter~5]{wajc2020matching}. 
We generalize the latter result to the more challenging \emph{edge arrivals} setting in \emph{general} graphs.
\begin{theorem}[See exact bounds in \cref{thm:ronding_matchings_theorem}]\label{thm:online-rounding}
    There exists an online $(1+o(1))$-approximate rounding algorithm for online matching $\vec{x}$ under adversarial edge arrivals, subject to the promise that $\max_e x_e \leq o(1)$.
\end{theorem}

We note that \Cref{thm:online_matching} is the special case of \Cref{thm:online-rounding} applied to the fractional matching assigning values $1/\Delta=o(1)$ to each edge. 
While we focus on this special case in the paper body for ease of exposition, we believe that our more general rounding algorithm is of independent interest and has broader applicability. Illustrating this, in \Cref{sec:unknown_delta}, we combine our rounding algorithm with a rounding framework and algorithm for fractional edge coloring of \cite{cohen2019tight} to obtain the first online edge coloring algorithm beating the naive greedy algorithm for online edge coloring under vertex arrivals with \emph{unknown} maximum degree $\Delta=\omega(\log n)$; specifically, we show (details in \cref{thm:unknown_delta}) that $(1.777+o(1))\Delta$-edge-colorings are attainable in this setting, approaching the lower bound of $1.606\Delta$ proved by \cite{cohen2019tight}.

\subsection{Related Work}\label{sec:related-work}

Since edge coloring is the problem of decomposing a graph into few matchings, it is natural to relate this problem to online matching.

The study of online matching was initiated by Karp, Vazirani and Vazirani \cite{karp1990optimal}, whose main result was a positive one: they presented an optimal algorithm under one-sided vertex arrivals in bipartite graphs, showing in particular that the greedy algorithm's competitive ratio is suboptimal for this problem.
Similar positive results were later obtained for several generalizations, including weighted matching \cite{aggarwal2011online,fahrbach2020edge,blanc2021multiway,gao2021improved}, budgeted allocation (a.k.a~AdWords) \cite{mehta2007adwords,huang2020adwords}, and fully-online matching \cite{huang2020fully,huang2020fully2}.
However, in the most general setting, i.e., under edge arrivals, the competitive ratio of the trivial greedy algorithm is optimal \cite{gamlath2019online}.

The study of online edge coloring was initiated by Bar-Noy, Motwani and Naor \cite{bar1992greedy}, who presented a negative result: they showed that the greedy algorithm is optimal, at least for low-degree graphs. 
Positive results were later obtained under random-order arrivals \cite{aggarwal2003switch,bahmani2012online}, culminating in a resolution of \Cref{conj:BMN} for such arrivals, using the nibble method \cite{bhattacharya2021online}.
For adversarial arrivals, \cite{cohen2019tight,blikstad2023simple} show that in bipartite graphs with one-sided vertex arrivals, the same conjecture holds. This was followed by progress in general graphs, under vertex arrivals \cite{saberi2021greedy}, and edge arrivals \cite{kulkarni2022online}, though using more than the hoped-for $(1+o(1))\Delta$ many colors. 
We obtain this bound in this work. Thus, we show that not only is the greedy algorithm suboptimal for online edge coloring, but in fact in the most general edge arrival setting, the online problem is asymptotically no harder than its offline counterpart.

%% file: Sections/prelims.tex
\section{Preliminaries}

\noindent\textbf{Notation.} As standard, we denote by $N(v)$ and $\delta(v)$ the neighborhood and edge sets of $v$, respectively, and denote the number of vertices and edges of $G$ by $n:=|V|$ and $m:=|E|$. We also denote by $\deg_H(v)$ the degree of vertex $v$ in (sub)graph $H$, and use the shorthand $\deg(v):=\deg_G(v)$. We say an event happens \emph{with high probability in a parameter $k$} if it happens with probability at least $1-k^{-c}$ for a constant $c>0$.

\paragraph{Problem definition and notation.} 
In the online problems studied in this paper, the input is an undirected simple graph $G := (V,E)$ with known maximum degree $\Delta$.
Its edges arrive one at a time, with edge $e_t\in E$ arriving at time $t$. 
An \emph{online edge coloring algorithm} must color each edge $e_t$ upon arrival with a color distinct from its adjacent edges.
Similarly, an \emph{online matching algorithm} must decide whether to match $e_t$ upon arrival, if none of its endpoints are matched.
For both problems, we consider randomized algorithms and assume that the input is generated by an oblivious adversary, which fixes the input graph and edges' arrival order before the algorithm receives any input.\footnote{By standard reductions \cite{ben1994power}, a result quantitatively similar to our main result against an \emph{adaptive} adversary would be equivalent to the task of finding a \emph{deterministic} algorithm.}
The objective of edge coloring algorithms is to output a coloring using as few colors as possible, close to the offline optimal 
$\Delta$ or $\Delta + 1$ colors~\cite{vizing1964estimate}.
The objective of online matching algorithms is traditionally to output a large matching. However, due to the reduction mentioned in the introduction, and restated more formally below, our interest will be in online matching algorithms that match each edge with high probability, close to $1/\Delta$.

\begin{lemma}[Reduction (\cite{cohen2019tight, saberi2021greedy})]
\label{lemma:reduction2}
    Let $\calA$ be an online matching algorithm that, on any graph of maximum degree $\Delta = \omega(\log n)$, matches each edge with probability at least $1 / (\alpha \cdot \Delta)$, for $\alpha \geq 1$. 
    Then, there exists an online edge coloring algorithm $\calA'$ that on any graph with maximum degree $\Delta = \omega(\log n)$ outputs an edge coloring with $(\alpha + O((\log n / \Delta)^{1/4})) \cdot \Delta$ colors with high probability in $n$.
\end{lemma}

In \Cref{sec:matching_peeling}, we generalize the above lemma, and use this generalization to obtain results for online \emph{list} edge coloring 
(each edge has a possibly distinct palette) and for
online \emph{local} edge coloring (each edge~$e$ should be colored with a color of index not much higher than $\max_{v\in e} \deg(v)$). In particular, the appendix implies (see \cref{lemma:reduction_improved_appendix}) that one can reduce the slack above to $(\alpha + O((\log n / \Delta)^{1/3})) \cdot \Delta$ colors.

\paragraph{Martingales.} A crucial ingredient in the analysis of our algorithms is the use of martingales.

\begin{definition}[Martingale] \label{def:martingales}
    A sequence of random variables $Y_0,\dots,Y_m$ is a \emph{martingale with respect to} another sequence of random variables  $X_1,\dots,X_m$ if the following conditions hold:
    \begin{itemize}
        \item $Y_k$ is a function of $X_1,\dots,X_k$ for all $k\geq 1$.
        \item $\E[ |Y_k|] < \infty $ for all $k \geq 0$.
        \item $\E[ Y_k \mid X_1,\dots,X_{k-1}] = Y_{k-1}$ for all $k \geq 1$.
    \end{itemize}
\end{definition}

The technical advantage of using martingales in our analysis is their amenability to specialized concentration inequalities which, unlike Chernoff-Hoeffding type bounds, do not require independence (or negative correlation) between the involved random variables. In particular, we will use a classic theorem due to Freedman providing a Chernoff-type bound only depending on the \emph{step size} and on the \emph{observed variance} of the martingale, with the latter defined as follows. For any possible outcomes $(x_1,\dots,x_{m-1})$ of the random variables $(X_1,\dots,X_{m-1})$, let:
\begin{equation*}
    W_m(x_1,\dots,x_{m-1}) := \sum_{i=1}^{m} \E[ (Y_i - Y_{i-1})^2 \mid X_1 = x_1,\dots,X_{i-1} = x_{i-1} ]
\end{equation*}
be the \emph{observed variance} encountered by the martingale on the particular sample path $x_1,\dots,x_{m-1}$ it took. To simplify notation, we usually assume that $x_1,\dots,x_{m-1}$ are chosen arbitrarily and write:
\begin{align*}
    W_m := \sum_{i=1}^{m} \E[ (Y_i - Y_{i-1})^2 \mid X_1,\dots,X_{i-1}].
\end{align*}
\begin{lemma}[Freedman's Inequality \cite{freedman1975tail}; see also {\cite[Theorem 12]{freedman}}, {\cite[Theorem 3.15]{habib1998probabilistic}}]\label{thm:freedman_inequality}
    Let $Y_0,\dots,Y_m$ be a martingale with respect to the random variables $X_1,\dots,X_m$. 
    If $|Y_k - Y_{k-1}| \leq A$ for any $k \geq 1$ and
    $W_m\leq \sigma^2$ always,
    then for any real $\lambda \geq 0$:
    \begin{equation*}
        \Pr[ |Y_n - Y_0| \geq \lambda ] \leq 2 \exp\left( - \frac{\lambda^2}{2(\sigma^2 + A \lambda / 3)} \right).
    \end{equation*}
\end{lemma}
We remark that we tailored the inequality to our use, and a  more general version holds~\cite{freedman1975tail}.

%% file: Sections/intuition.tex
\section{Online Matching Algorithm} \label{sec:intuition}
In this section, we design an online matching algorithm as guaranteed by \cref{thm:online_matching}.

\subsection{Our First Matching Algorithm}
\begin{figure}[!b]
    \centering
    \includegraphics{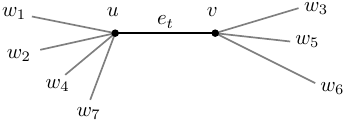}
    \caption{An example of the neighborhood of $e_t = (u,v)$ with $k = 7$. }
    \label{fig:local_randomness}
\end{figure}

 We first describe our key modification of the basic algorithm presented in the introduction. 
 Analogous to that algorithm, upon arrival of edge \( e_t \) whose endpoints are still unmatched, we match $e_t$ with a ``scaled'' probability \( P(e_t) \).
 However, and crucially for our analysis, our scaling factor will depend on the specific execution (random choices) of the algorithm.
 To illustrate this modification, refer to Figure~\ref{fig:local_randomness}, which depicts the neighborhood surrounding \( e_t \). Here, we denote by \( e_{t_j} \) the $j$-th edge connecting a vertex in $e_t$ with a vertex \( w_j \), with these \( k=7 \) edges appearing sequentially before \( e_t \), with \( t_1 < t_2 < \cdots < t_k < t\). 
 Suppose now that we fixed the randomness associated with all edges other than $\{e_t,e_{t_1,},e_{t_2},\dots,e_{t_k}\}$, and conditioned on this event, which we refer to as $R$, let \( P_R(e_{t_j}) \) represent the probability that the algorithm will add the edge \( e_{t_j} \) to the matching, assuming that none of the preceding edges \( e_{t_1}, e_{t_2}, \ldots, e_{t_{j-1}} \) have been matched.
  Consequently, we have:
\begin{gather*}
    \Pr[u,v \textrm{ both unmatched until time $t$}\mid R] = \prod_{j=1}^k(1- P_R(e_{t_j})),
\end{gather*}
and overall, by total probability,
\begin{gather*}
    \Pr[u,v \textrm{ both unmatched until time $t$}] = \sum_R \Pr[R]\cdot  \prod_{j=1}^k(1- P_R(e_{t_j}))\,.
\end{gather*}

The primary complication with the initial algorithm in the introduction
is the intricate correlation within the joint distribution of $P_R(e_{i_1})$, $ \ldots, P_R(e_{i_k})$ as  functions of the randomness $R$ of edges outside the direct neighborhood of $e_t$.
This correlation complicates both the computational aspect of determining the probability that vertices \( u \) and \( v \) are unmatched until time \( t \), as well as the theoretical analysis of the algorithm's competitive ratio.

Our algorithm overcomes this challenge with a, in hindsight, simple strategy. We utilize a scaling factor conditional on the randomness \( R \), i.e., we scale with respect to the ``observed'' probabilities, thus ensuring that the resulting online algorithm is both computationally efficient and (as we will see) theoretically tractable to analyze. Specifically, we obtain the following algorithm:
\begin{center}
\begin{minipage}{0.95\textwidth}
\begin{mdframed}[hidealllines=true, backgroundcolor=gray!15]
\begin{algorithm}[\textsc{NaturalMatchingAlgorithm}] \ \\[0.2cm]
When an edge $e_t = (u,v)$ arrives, match it with probability  \[
P(e_t) \gets  
\begin{cases}
    \frac{1}{\Delta + q}\cdot \frac{1}{ \prod_{j=1}^k(1- P(e_{t_j}))} & \mbox{if $u$ and $v$ are still unmatched,} \\
    0 & \mbox{otherwise,}
\end{cases}
\]
where $e_{t_1}, \ldots, e_{t_k}$ are those previously-arrived edges incident to the endpoints of $e_t$.
\label{alg:natural-alg}
\end{algorithm}
\end{mdframed}
\end{minipage}
\end{center}

Note that the values $P(e_{t_j})$ needed to compute $P(e_t)$ are all defined at time $t$ (and easy to compute), since any such edge $e_{t_j}$ arrived before $e_t$. 
Moreover, assuming $u$ and $v$ are unmatched (\emph{free}), these values equal $P_R(e_{t_j})$, where $R$ is the event corresponding to the random bits used in this execution for the edges outside the neighborhood of $e_t$. Thus, if  $P(e_t) \leq 1$ for every edge $e_t$, this algorithm is well-defined, and attains the right marginals, by total probability over $R$:
\begin{align*}
    \Pr[\mbox{$e_t$ matched}] &= \sum_{R} \Pr[R] \cdot P_R(e_t) \cdot \Pr[u,v \textrm{ both free until time $t$}\mid R] = \sum_R \Pr[R]\cdot  \frac{1}{\Delta + q} = \frac{1}{\Delta + q}\,.
\end{align*}

However, our new natural algorithm is ill-defined, as $P(e_t)$ may exceed $1$, even on trees.
 For example, suppose edges in \Cref{fig:prob_more_than_one} arrive in a bottom-up, left-to-right order, and no edge before $e_{t}$ is matched by the algorithm (this is a very low-probability event). 
 Then $P(e_{t_1}) = \frac{1}{\Delta+q}\cdot \frac{1}{\prod_{i=1}^{\Delta-2} (1-P(e_i))}$, where $e_1, \ldots, e_{\Delta-2}$ are the edges below $e_{t_1}$.
But the term $\prod_{i} (1-P(e_i))$ is the probability that none of the edges $e_i$ are matched by the algorithm, which is exactly $1-\frac{\Delta-2}{\Delta+q} = \frac{q+2}{\Delta+q}$, since each $e_i$ in this example is matched with probability exactly $\frac{1}{\Delta+q}$ and these are disjoint events. Hence $P(e_{t_1}) = \frac{1}{\Delta+q}\cdot \frac{\Delta+q}{q+2} = \frac{1}{q+2}$. We can calculate $P(e_{t_2})$ in the same way, remembering to also scale up by $\frac{1}{1-P(e_{t_1})}$, and we get $P(e_{t_2}) = \frac{1}{q+2}\cdot \frac{1}{1-P(e_{t_1})} = \frac{1}{q+1}$.  Continuing in this fashion, $P(e_{t_i}) = \frac{1}{q+3-i}$ for any $i=1,\ldots, q+1$ (importantly $P(e_{t_i}) < 1$, so indeed with some non-zero probability the algorithm will not match any of them).
Finally, by a similar calculation, $P(e_t) = \frac{1}{q+1} \cdot \frac{1}{\prod_{i=1}^{q+1}(1-P(e_{t_i}))} = \frac{q+2}{q+1} > 1$, making the algorithm undefined. 
However, in \emph{most} runs of the algorithm, many bottom-level edges are matched, and so $P(e_{t_i})=0$ for many $i$, and $P(e_t)$ is much smaller.

\begin{figure}[!ht]
    \centering
    \includegraphics[scale=1]{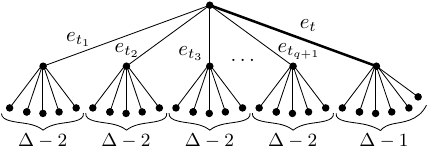}
    \caption{Example where \cref{alg:natural-alg} might be undefined.}
    \label{fig:prob_more_than_one}
\end{figure}

In the above example, the event that $P(e_{t})$ in this case exceeds one hinges on low-probability events. This therefore does not rule out matching probabilities of, say, $1/(\Delta+q)-1/\Delta^3 \geq 1/(\Delta+O(q))$, so long as we avoid the use of $P(e_t)$ as probabilities when $P(e_t)>1$.
In the next section we present a modification of \Cref{alg:natural-alg} doing just this, and provide an overview of its analysis.

\subsection{The Analysis-Friendly Matching Algorithm}

As discussed, it is not at all clear whether the random variables $P(e_t)$ in \cref{alg:natural-alg}, which we interpret as probabilities, even \emph{are} valid probabilities, i.e., whether they are upper bounded by $1$.
To avoid working with potentially invalid probabilities, we use a slightly different variant of \cref{alg:natural-alg},  whose pseudocode is given by \cref{alg:edge-arrival}.
This variant not only addresses the ill-defined probability concern, but also introduces a more precise notation, facilitating our analysis. However, this increased precision might initially obscure the connection between \cref{alg:edge-arrival} and the more intuitive \cref{alg:natural-alg}.

To clarify this relationship, consider the scaling factor ${1}/{ \prod_{j=1}^k(1- P(e_{t_j}))}$ for an edge $e_t = (u,v)$, which can be partitioned based on the edges incident to vertices $u$ and $v$. Let $\delta_{t}(u)$ denote the set of edges incident at $u$ and arriving before $e_t$. We then define $F_t(u)$ (and similarly $F_t(v)$) as follows:
\begin{align*}
    F_t(u) := \prod_{e_{t_j} \in \delta_t(u)} (1 - P(e_{t_j})) \qquad \textrm{and} \qquad 
    F_t(v) := \prod_{e_{t_j} \in \delta_t(v)} (1 - P(e_{t_j})).
\end{align*}
Consequently, as $G$ is a simple graph, the probability $P(e_t)$ used in \cref{alg:natural-alg} can be reformulated as:
\[
P(e_t) = 
\begin{cases}
    \frac{1}{\Delta + q}\cdot \frac{1}{ F_t(u) \cdot F_t(v)} & \text{if $u$ and $v$ are still unmatched by time $t$}, \\
    0 & \text{otherwise}.
\end{cases}
\]
The above aligns with the pseudo-code in \cref{alg:edge-arrival}. 
Assuming that $\hat P (e_t) = P(e_t)$ (and using that $G$ is a simple graph and hence has no parallel edges), it becomes evident that \cref{alg:edge-arrival} is equivalent to \cref{alg:natural-alg} under the premise that all probabilities $P(e_t)$ are at most one. The critical modification in \cref{alg:edge-arrival} is the introduction of $\hat P$, possibly distinct from $P$, ensuring that $P(e_t) \leq 1$. This is accomplished by constraining the values of $F_t(v)$, now redefined in terms of $\hat P$, i.e., $F_t(v) := \prod_{e_{t_j} \in \delta_t(v)} (1 - \hat P(e_{t_j}))$, to not fall below $\frac{q}{4\Delta}$, implying $P(e_t)\leq 1$ for appropriately small $q$ (see \cref{obs:Ftv_and_Pet_bounds}).

\begin{figure}[!ht]
\begin{center}
\begin{minipage}{0.95\textwidth}
\begin{mdframed}[hidealllines=true, backgroundcolor=gray!15]

\begin{algorithm}[\edgearrivalalg] \ \\[0.2cm]
\textbf{\emph{Initialization:}} Set $F_1(v) \gets 1$ for every vertex $v$ and $M_1 \gets \emptyset$.  \\[0.1cm]

\emph{At the arrival of edge $e_t = (u,v)$ at time $t$:}
\begin{itemize}
    \item  Sample $X_t \sim \textrm{Uni}[0,1]$.
    \item Define 
    \begin{align*}
        P(e_t) & = \begin{cases}
         \frac{1}{\Delta + q} \cdot 
         \frac{1}{F_t(u) \cdot F_t(v)} 
         & \mbox{if $u$ and $v$ are unmatched in $M_t$,} \\ 
         0 & \mbox{otherwise.}
         \end{cases}
         \intertext{and}
        \hat P(e_t) & = \begin{cases}
        P(e_t) & \mbox{if $\min\{F_t(u), F_t(v)\}\cdot (1-P(e_t))\geq q/(4\Delta)$} \\ 
         0 & \mbox{otherwise.}
         \end{cases} 
    \end{align*}
    
    \item Set 
    \begin{itemize}
        \item $F_{t+1}(u) \leftarrow F_t(u) \cdot (1- \hat P(e_t))$;
        \item $F_{t+1}(v)\leftarrow F_t(v) \cdot (1-\hat P(e_t))$; 
        \item $M_{t+1} \gets \begin{cases} M_t \cup \{e_t\} & \mbox{ if $X_t < \hat P(e_t)$,} \\
        M_t & \mbox{otherwise.}
\end{cases}$
    \end{itemize}
\end{itemize}
\label{alg:edge-arrival}
\end{algorithm}
\end{mdframed}
\end{minipage}
\end{center}
\end{figure}

\paragraph{Analysis: Intuition and Overview.}
As in \cref{alg:natural-alg}, it is not hard to prove that edge $e_t$ is selected with a probability of $1/(\Delta+q)$ if $P(e_t)$ always equals $\hat P(e_t)$, as detailed in \cref{subsec:analysis_part1}. Therefore, the meat of our analysis, in \cref{subsec:analysis_part2}, focuses on proving that for any edge $e_t$, the equality $P(e_t) = \hat P(e_t)$ holds with high probability in $\Delta$.
For intuition why this should be true, note that if each of the $\deg_t(v)$ edges of vertex $v$ by time $t$ is matched with probability $\frac{1}{\Delta + q}$, then the value $F_t(v)$---that intuitively stands for the probability of $v$ being free at time~$t$---should be:
\begin{equation*}
    F_t(v) = 1 - \frac{\deg_t(v)}{\Delta + q} \geq \frac{q}{\Delta + q} \approx \frac{q}{\Delta}.
\end{equation*}
Above, the inequality follows from $\deg_t(v) \leq \Delta$ and the approximation follows by our choice of $q = o(\Delta)$. 
Moreover, basic calculations (see \cref{lemma:sufficient_cond_for_Phat_equal_P}) imply that $\hat P(e_t) = P(e_t)$  if $\min\{F_t(u),F_t(v)\} \geq \frac{q}{3\Delta}$. In other words, $\hat P(e_t) \neq P(e_t)$ only if the $F_t(\cdot)$-value has dropped significantly below its ``expectation'' for one of the endpoints of $e_t$.

The core of the analysis then boils down to proving concentration bounds that imply, for any vertex $v$ and time $t$, that $F_t(v) \geq \frac{\Delta}{3q}$ with high probability in $\Delta$.
The values $F_t(v)$ are non-increasing as $t$ grows, so it suffices to prove this inequality for the final value $F(v):=F_m(v)$. Let $u_1,\dots,u_\ell$ be the neighbors of $v$ in the final graph. By simple calculations (\Cref{lemma:main_ineq}), we  show that:
\begin{equation} \label{eq:intuition_lowerbound_Fv}
    F(v) \geq 1 - \sum_{i = 1}^\ell E_i \cdot \frac{1}{\Delta + q} \cdot \frac{1}{F_{t_i}(u_i)},
\end{equation}
where the binary random variables $E_i$ are non-zero if and only if the neighbor $u_i$ of $v$ is unmatched by the time $t_i$ at which the edge $(v,u_i)$ arrives. If we denote by $S := \{u_i \in N(v) \mid \text{ $u_i$ not matched until time $t_i$}\}$, 
then $E_i := \mathbb{1}[u_i \in S]$.

The expectation of the sum $Y := \sum_{i=1}^\ell E_i \cdot \frac{1}{\Delta + q} \cdot \frac{1}{F_{t_i}(u_i)}$ can be shown to be at most $\frac{\Delta}{\Delta + q}$. If the variables $E_i$ and $F_{t_i}(u_i)$ were independent, one could now use Chernoff-Hoeffding type bounds to conclude that $Y \leq \frac{\Delta}{\Delta + q/2}$ with high probability in $\Delta$, proving $F(v) \geq 1 - Y_m \geq \frac{q}{3\Delta}$ in the process (see \Cref{lemma:main_ineq}). However, in general, the events of different neighbors $u_i$ of $v$ being matched when $(v,u_i)$ arrives are not independent, and so the variables $E_i, F_{t_i}(u_i)$ are correlated, making such an approach not applicable. 
We overcome this by interpreting the right-hand side of \eqref{eq:intuition_lowerbound_Fv} as the final state of a martingale.

Concretely, our main idea is to parameterize the set $S$, the random variables $E_i$ and the sum $Y$ over time. Assuming that the input stops at time step $t \leq m$, one can naturally define the analogues of $S$, $E_i$, and $Y$ up to time $t$ by
$S_t := \{u_i \in N(v) \mid \text{ $u_i$ not matched until time $\min\{t,t_i\}$}\}$, $E_{ti} := \mathbb{1}[u_i \in S_t]$ and: \begin{equation*}
    Y_t := \sum_i^\ell E_{ti} \cdot \frac{1}{\Delta + q} \cdot \frac{1}{F_{\min\{t,t_i\}}(u_i)}.
\end{equation*}
So, at each time step, we either change the value of a term or drop a term.
With this notation, we trivially have that $S = S_m$, $E_i = E_{mi}$ and $Y = Y_m$.
The advantage of this representation is that $Y_0,\dots,Y_m$ turns out to be a martingale with respect to the random variables $X_1,\dots,X_m$ sampled by \Cref{alg:edge-arrival}. As $Y_0 \leq \frac{\Delta}{\Delta + q}$, our objective reduces to proving that the following holds with high probability in $\Delta$:
\begin{equation*} \label{eq:intuition_concentration_martingale}
    |Y_m - Y_0| \leq \frac{\Delta}{\Delta + q/2} - \frac{\Delta}{\Delta + q}.
\end{equation*}
As the martingale $Y_0,\dots,Y_m$ takes $\Delta^2$ non-trivial steps (based on the two-hop neighborhood), it is not enough to use the maximum step size and the number of steps to argue about concentration as done in, e.g., Azuma's inequality. However, we can bound the maximum step size and the observed variance, which is sufficient for applying Freedman's inequality (\Cref{thm:freedman_inequality}), yielding our desired result. In the next section we substantiate the above intuition, and analyze \Cref{alg:edge-arrival}.

%% file: Sections/analysis.tex
\section{Analysis of the Online Matching Algorithm} \label{sec:analysis}

In this section we present the formal analysis of \cref{alg:edge-arrival}, and prove that it matches each edge $e_t$ with probability  at least $1/(\Delta+O(q))$, with $q$ to be chosen shortly.
Our analysis is 
divided into two parts. 

In the first part (\cref{subsec:analysis_part1}), we prove that if $\hat P(e_t) = P(e_t)$, i.e., the values $F_t(v)$ for both $v\in e_t$ are large enough, then we match $e_t$ with probability at least $1/({\Delta + q})$ (\cref{lemma:sufficient_cond_for_right_marginal}).

In the second part (\cref{subsec:analysis_part2}), we remove the assumption that $P(e_t)=\hat{P}(e_t)$ and prove that \Cref{alg:edge-arrival} achieves a matching probability of at least $1/(\Delta+4q)$, by showing that with high probability in $\Delta$, the values $F_t(v)$ are large enough to guarantee  $P(e_t)=\hat{P}(e_t)$. To prove the latter high-probability bound, we interpret a sufficient desired lower bound (\cref{lemma:main_ineq}) as the final state of a martingale, and use Freedman's inequality (\cref{thm:freedman_inequality}) to prove that $F_t(v)$ is likely sufficiently large for our needs.

\paragraph{Choice of $q$.} We will use $q := \sqrt{200} \cdot \Delta^{3/4} \ln^{1/2}\Delta$, for reasons that will become clear in the proof of \Cref{lemma:bound_on_Y}. For the rest of the section we will only make use of the following corollaries of our choice of $q$:
\begin{equation} \label{eq:assumption}
    8 \sqrt{\Delta} \le q \le \Delta/4.
\end{equation}
Note that the upper bound on $q$ not only follows from its choice (for sufficiently large $\Delta$), but we may  also assume this bound without loss of generality: if $q>\Delta/4$, then simply picking a random color used by the $(2\Delta-1)$ colors of the greedy online coloring algorithm will match each edge with probability $1/(2\Delta-1)$, greater than our desired $1/(\Delta+4q)$ matching probability.

\subsection{A Sufficient Condition for \texorpdfstring{$1/(\Delta+q)$}{1/(Δ+q)} Matching Probability} \label{subsec:analysis_part1} 

We begin with a simple observation, which will facilitate our characterization of random values associated with \Cref{alg:edge-arrival} under various conditionings.

\begin{observation} \label{obs:fpp-det-by-M}
    For any time $t$, the random variables $F_{t}(v), P(e_t), \hat P(e_t)$ are determined by the current partial input $e_1,\dots,e_t$ and the current matching $M_{t-1}$.
\end{observation}
\begin{proof}
   Since $P(e_t)$ and $\hat P(e_t)$ are determined by the values of the variables $F_t(v)$, it suffices to prove the statement only for these latter variables.
    This follows by induction on $t$. For the base case, we have $F_1(v) = 1$ for all vertices, which implies the statement trivially. For the inductive step with $t > 1$, note that by construction $F_t(v)$ is determined by the values $\{\hat P(e_{t'}) : \text{$t'<t$ and $e_{t'} =(u', v)$ is incident to $v$}\}$. Any such value $\hat P(e_{t'})$ is in turn a function of $F_{t'}(u'), F_{t'}(v)$  and $P(e_{t'})$. By the inductive hypothesis, $F_{t'}(u'), F_{t'}(v)$ are functions of $M_{t'-1}$ and therefore also functions of $M_{t-1}$ (since $t'< t$). Finally, the value $P(e_{t'})$ is determined by $F_{t'}(u'), F_{t'}(v)$ and $M_{t'-1}$, which are again functions of $M_{t-1}$.
\end{proof}

The following lemma formalizes the intuition discussed in \cref{sec:intuition}; it proves that \cref{alg:edge-arrival} has the correct behavior for an edge $e_t = (u,v)$ by assuming that the randomness outside the 1-neighborhood ($\delta(u)\cup \delta(v)$) of $e_t$ is fixed and $\hat P(e_t) = P(e_t)$.

\begin{lemma} \label{lemma:sufficient_cond_for_right_marginal}
    For any edge $e_t = (u,v)$ it holds that $$\Pr[X_t < P(e_t)] = \frac{1}{\Delta + q}.$$
\end{lemma}
\begin{proof}
    Fix all randomness except for the edges incident to $u$ and $v$. 
    That is, we fix the outcomes $X_{t'} = x_{t'}$ for all edges $e_{t'}\not \in \delta(u) \cup \delta(v)$.
    Denote this event by $A(\vec{x})$.
    Let $t_1 < \cdots < t_\ell$ be the  arrival times of the edges in $\delta(u) \cup \delta(v)$ before time $t$.
    The only randomness left up to the point of $e_t$'s arrival is now given by the random variables $X_{t_1}, \ldots, X_{t_\ell}$, which are \emph{independent} of $A(\vec{x})$.
    For the selection of $e_t$ to be possible, we need to condition on the event that none of the edges $e_{t_1}, \ldots, e_{t_\ell}$ are taken in the matching. We note that conditioning on $A(\vec{x})$ and $e_{t_1}, \ldots, e_{t_\ell} \not \in M_t$ completely determines $M_{t'}$ for all time steps $t'\leq t$. Using \Cref{obs:fpp-det-by-M}, we thus have that $\hat P(e_{t_1}), \ldots, \hat P(e_{t_\ell})$, and $F_t(u), F_t(v)$ are uniquely determined under this conditioning.  Let $\hat p(e_{t_1}), \ldots, \hat p(e_{t_\ell})$, and $f_t(u), f_t(v)$ be the concrete values of these random variables under this conditioning.
    We then have:
     \begin{align*}
      \Pr[ \{e_{t_1}, \ldots, e_{t_\ell}\} \cap M_t = \emptyset  \mid A(\vec{x})] & = \prod_{i=1}^\ell \Pr[X_{t_i} > \hat p(e_{t_i})] = \prod_{i=1}^\ell (1- \hat p(e_{t_i})) = f_t(u) \cdot f_t(v).
     \end{align*}
    Above, for the first equality we used the independence of the different $X_{t_i}$, while for the last equality we partitioned the edges incident to $u$ and $v$ to obtain the factors $f_t(u)$ and $f_t(u)$ respectively. Note that, because the graph is simple and so parallel edges are disallowed, this partitioning is well defined, as none of the edges $e_{t_i}$ connects $u$ to $v$. 
    If $u$ or $v$ are unmatched before time $t$, we get: $P(e_t)=\frac{1}{\Delta+q}\cdot\frac{1}{f_t(u)\cdot f_t(v)}$. Since the random variable $X_t$ is independent from $M_t$, the above then yields the desired equality when conditioning on $A(\vec{x})$:
    \begin{align*}
        \Pr[ X_t < P(e_t)   \mid A(\vec{x})]  & = 
          \Pr[ X_t < P(e_t) \mbox{ and }  \{e_{t_1}, \ldots, e_{t_\ell}\} \cap M_t = \emptyset  \mid A(\vec{x})] \\
        & = \frac{1}{\Delta+q}\cdot\frac{1}{f_t(u)\cdot f_t(v)}\cdot f_t(u)\cdot f_t(v) \\
        & = \frac{1}{\Delta+q}\;.
    \end{align*}
    The lemma now follows by the law of total probability over all possible values of $A(\vec{x})$.
\end{proof}

Note that a consequence of the above lemma is that, if the values of $P$ and $\hat P$ were to \emph{always} coincide during the execution of \cref{alg:edge-arrival}, then all edges would be matched with probability $1/(\Delta+q)$. In the following section, we prove that the equality $\hat P(e_t) = P(e_t)$ holds with high probability in $\Delta$ for any edge~$e_t$, which, by simple calculations, implies a matching probability of $1/(\Delta+O(q))$.

\subsection{Analysis of \texorpdfstring{\Cref{alg:edge-arrival}}{Section \ref{alg:edge-arrival}} in General} \label{subsec:analysis_part2}

In this section we wish to prove that $\hat P$ and $P$ coincide for each edge with high probability in $\Delta$. 
This requires proving that $F_t(v)$ is likely to be high for all times $t$ and vertices $v$---intuitively that means that the probability that $v$ is unmatched is never too low.
We start by observing a trivial lower bound on $F_t(v)$ and upper bound on $P(e)$ that follows directly from the algorithm's definition.

\begin{observation}\label{obs:Ftv_and_Pet_bounds}
    $F_t(v) \geq q/(4\Delta)$ and $\hat P(e_t) \leq P(e_t) \leq 1/4$ for all vertices $v\in V$ and times~$t$.
\end{observation}
\begin{proof}
    For any fixed $v$, we prove the first statement by induction on $t$. For $t = 0$, using \eqref{eq:assumption}, we have $F_t(v) = 1 \geq q/(4\Delta)$.
    Now, assume that $F_t(v) \geq q/(4\Delta)$ for some $t \geq 0$. If either $v \notin e_t$ or $\hat P(e_t) = 0$, then clearly $F_{t+1}(v) = F_t(v)$ and the statement is proven. Otherwise, by the definition of $\hat P$ it follows that $F_{t+1}(v) = F_t(v) \cdot (1 - \hat P(e_t))=F_t(v) \cdot (1 - P(e_t)) \geq q/(4\Delta)$,  and again the statement is proven. The fact that $P(e_t) \leq 1/4$ is now a consequence of the previously proven fact and $q\geq 8\sqrt{\Delta}$, by \eqref{eq:assumption}:
    \begin{align*}
        P(e_t) & = \frac{1}{\Delta + q} \cdot \frac{1}{F_t(u) \cdot F_t(v)} \leq \frac{1}{\Delta + q } \cdot \frac{1}{(q/(4\Delta))^2} \leq \frac{16\Delta}{q^2} \leq \frac{1}{4}. \qedhere
    \end{align*} 
\end{proof}

To prove that $\hat P(e_t) = P(e_t)$ with high probability in $\Delta$, we note that by their construction in \cref{alg:edge-arrival}, the values $\hat P(e_t)$ can only differ from $P(e_t)$ if the values of the variables $F_t(v)$ are too small, and in particular are close to their lower bound of $q/(4\Delta)$ guaranteed by \Cref{obs:Ftv_and_Pet_bounds}.

\begin{observation} \label{lemma:sufficient_cond_for_Phat_equal_P}
    If $e_t = (u,v)$ and $\min\{F_{t}(u),\ F_{t}(v)\} \geq q/(3\Delta)$, then $\hat P(e_t) = P(e_t)$.
\end{observation}
\begin{proof}
    As $P(e_t) \leq 1/4$ by \cref{obs:Ftv_and_Pet_bounds}, we have $\min\{F_{t}(u), F_{t}(v)\}\cdot (1-P(e_t))  \geq (q/(3\Delta)) \cdot (3/4) = q/(4\Delta)$, which, by the algorithm's definition, in turn implies that $\hat P(e_t)=P(e_t)$.
\end{proof}

Given \Cref{lemma:sufficient_cond_for_Phat_equal_P} and \Cref{lemma:sufficient_cond_for_right_marginal}, it suffices to show that the probability that $F_t(v) < q/(3\Delta)$ is very small for all time steps~$t$. As $F_{t+1}(v) \leq F_t(v)$, it is thus sufficient to bound this probability at the very last step, i.e., to bound the probability that $F(v) < q/(3\Delta)$, where $F(v) := F_m(v)$. In the following lemma, we identify a sufficient condition---\Cref{eq:main_inequality}--- for the condition $F(v)\ge q/(3\Delta)$ (and thus also $\hat{P}(e) = P(e)$) to hold.
In particular, we lower bound $F(v)$ by only focusing on the impact of neighbors of $v$ on $P(e_{t_i})$ for edges $e_{t_i}\ni v$ and then applying the union bound.

\begin{lemma}
\label{lemma:main_ineq}
    Let $e_{t_1} = (u_1, v),\ldots, e_{t_\ell} = (u_\ell, v)$ be the edges incident to $v$, arriving at times $t_1< \cdots < t_\ell$. Let $S := \{u_i \in N(v) \mid u_i \not \in M_{t_i}\}$ be those neighbors $u_i$ that are not matched before time $t_i$ when the edge $e_{t_i} = (u_i, v)$ arrives. Then,
    \begin{align*}
        F(v) \ge 1-\sum_{u_i \in S} \frac{1}{\Delta + q}\frac{1}{F_{t_i}(u_i)}.
    \end{align*} 
    As a consequence, 
    $F(v) \ge q/(3\Delta)$ holds if \begin{align}
        \sum_{u_i \in S} \frac{1}{\Delta + q}\frac{1}{F_{t_i}(u_i)} \leq \frac{\Delta}{\Delta + q/2}.
        \label{eq:main_inequality}
    \end{align} 
\end{lemma}
\begin{proof}
We look at how $F_{t}(v)$ develops throughout the run of \cref{alg:edge-arrival}. When edge $e_{t_i} = (u_i, v)$ arrives, the algorithm sets $F_{t_i+1}(v) \gets F_{t_i}(v)\cdot (1-\hat{P}(e_{t_i}))$, yielding the following lower bound on $F_{t+1}(v)$, relying on $\hat{P}(e_{t_i})\leq P(e_{t_i})$: 
\begin{align*}
    F_{t_i+1}(v) \geq F_{t_i}(v)\cdot (1-P(e_{t_i}))&\ge F_{t_i}(v)\cdot \left(1-\frac{1}{\Delta+q}\frac{1}{F_{t_i}(v)F_{t_i}(u)}\right)
=
F_{t_i}(v) - \frac{1}{\Delta+q}\frac{1}{F_{t_i}(u)}.
\end{align*}
Above, the second inequality is only an equality if both $v,u_i$ are not matched before time $t_i$, and in particular $u_i\in S$.
In the alternate case, we have that $\hat{P}(e_{t_i})=P(e_{t_i})=0$, and so $F_{t+1}(v)=F_t(v)$. We conclude that $F_{t_i+1}(v)-F_{t_i}(v)$ can only be non-zero for times $t_i$ with $u_i$ previously unmatched ($u_i\in S$), in which case $F_{t_i+1}(v)-F_{t_i}(v)\geq -\frac{1}{\Delta+q}\frac{1}{F_{t_i}(u_i)}$.
Since $F_1(v) = 1$ initially, the first part of the lemma follows by summing over all $u_i\in S$. 

The second part of the lemma, whereby $F(v)\geq q/(3\Delta)$ provided \Cref{eq:main_inequality} holds, 
 now follows from the first part and a simple calculation, using that $q\leq \Delta/4$ by \Cref{eq:assumption}:
    \begin{equation*}
        F(v) \geq 1 - \sum_{j=1}^k \frac{1}{\Delta + q}\frac{1}{F_{t_{i_j}}(u_{i_j})}
         \stackrel{\eqref{eq:main_inequality}}{\geq} 1 - \frac{\Delta}{\Delta + q/2} = \frac{q/2}{\Delta + q/2} \geq q/(3\Delta)\,. \qedhere
    \end{equation*} 
\end{proof}

Similar arguments to those in \Cref{subsec:analysis_part1} can be used to prove that the sum $P:=\sum_{u_i \in S} \frac{1}{\Delta + q}\frac{1}{F_{t_i}(u_i)}$ satisfies $\E[P]\leq \frac{\Delta}{\Delta+q}$. (This also follows from the subsequent martingale analysis later.)
That the expectation of $P$ is bounded away from the upper bound of $\frac{\Delta}{\Delta+q/2}$ required in \eqref{eq:main_inequality} hints at using appropriate concentration inequalities, such as Chernoff-Hoeffding type bounds, to prove $P \leq \frac{\Delta}{\Delta+q/2}$ with high probability in $\Delta$.
Unfortunately, $P$ is not a sum of independent or negatively correlated random variables in general, and therefore Chernoff-Hoeffding bounds are not applicable. Instead, in the subsequent sections we model the development of $P$ as a martingale process, which will allow proving the desired concentration inequality without having to argue explicitly about correlations.

\subsubsection{Our martingale process}

 Fix a vertex $v$.
To prove that the sufficient condition $F(v)\geq q/(3\Delta)$ of inequality \eqref{eq:main_inequality} holds often in general, we view the development of the left-hand-side of \eqref{eq:main_inequality} as a martingale. For a time step $t$, define:
\[
    S_t := \{ u_i \in N(v) \mid u_i \not \in M_{\min\{t, t_i\}}\} \quad \mbox{and} \quad Y_{t-1} := \sum_{u_i \in S_t} \frac{1}{\Delta + q}\frac{1}{F_{\min\{t,t_i\}}(u_i)}.
\]
Recall that $t_i$ is the time step at which the edge $e_{t_i}=(v,u_i)$ arrives, and $N(v)$ is the final set of neighbors of $v$ in the graph. Hence, $S_t$ contains all neighbors of $v$ in the final graph (including the future neighbors), except those neighbors that were already matched by the time $\min \{t,t_i\}$. In particular, if $u_i \in S_{t_i}$, i.e., $u_i$ was not matched by the time $t_i$ it gets connected to $v$, it will remain inside all future sets $S_t$, for $t \geq t_i$. Also, notice that both $S_t$ and $Y_t$ are unknown to the algorithm at time $t$, as their definition requires ``future'' knowledge of the input graph, and that they are only used for the analysis.

With the above notation, $Y_0 = \frac{\deg(v)}{\Delta+q} \le \frac{\Delta}{\Delta + q}$ and $Y:=Y_{m}$ equals the left-hand-side of \Cref{eq:main_inequality}. $Y_{t-1}$ is determined by the independent random variables $X_1, \ldots, X_{t-1}$ sampled by \cref{alg:edge-arrival}. As we now show, $Y_0,\dots,Y_m$ indeed form a martingale:
\begin{lemma} \label{lemma:proof_y_is_martingale}
    $Y_0, \dots, Y_m$ form a martingale w.r.t.\ the random variables $X_1,\dots,X_m$. Furthermore, the difference $Y_t - Y_{t-1}$ is given by the following two cases:
    \begin{itemize}
        \item If $e_t$ is added to $M_{t+1}$, which happens with probability $\hat P(e_t)$, then:
        \begin{equation} \label{eq:diff_y_case1}
            Y_{t} - Y_{t-1}  = - \frac{1}{\Delta + q} \sum_{u_i \in S_t \cap e_t} \frac{1}{F_{t}(u_i)}.
        \end{equation}
        \item If instead $e_t$ is not added to $M_{t+1}$, which happens with probability $1 - \hat P(e_t)$, then:
        \begin{equation} \label{eq:diff_y_case2}
            Y_{t} - Y_{t-1}  = \frac{1}{\Delta + q} \cdot \frac{\hat P(e_t)}{1 - \hat P(e_t)} \sum_{u_i \in S_t \cap e_t} \frac{1}{F_{t}(u_i)}.
        \end{equation}
    \end{itemize}
\end{lemma}
\begin{proof}
    To prove the martingale property, we check the conditions given by \Cref{def:martingales}. First, notice that fixing $X_1,\dots,X_t$ in \Cref{alg:edge-arrival} determines the set $S_t$ and the values of the random variables $F_{\min\{t,t_i\}}(u_i)$ used to define $Y_t$. Hence, $Y_t$ is a function of $X_1,\dots,X_t$. As $F_{\min\{t,t_i\}}(u_i) \geq q/(4\Delta)$, obviously $\E[Y_t] < \infty$.

    It remains to show that $\E[Y_t \mid X_1,X_2,\dots,X_{t-1}] = Y_{t-1}$. We first verify the claimed identities at \eqref{eq:diff_y_case1} and \eqref{eq:diff_y_case2}. If the edge $e_t$ arriving at time $t$ is not incident to any $u_i \in S_t$ with $t<t_i$, then $Y_{t+1} = Y_t$ deterministically, and \eqref{eq:diff_y_case1}, \eqref{eq:diff_y_case2} hold as their right hand sides are equal to $0$. On the other hand, if $e_t$ is incident to one or two such vertices $u_i \in S_t$, then we have two cases: $e_t$ might be added to $M_{t+1}$ or not. If $e_t$ was matched, $S_{t+1} = S_t \setminus (S_t\cap e_t)$, so some terms are dropped from the sum and the identity \eqref{eq:diff_y_case1} follows. Otherwise, if $e_t$ was not matched, we have for any $u_i \in S_t \cap e_t$:
    \begin{align*}
        \frac{1}{\Delta + q}\left(\frac{1}{F_{t+1}(u_i)} -\frac{1}{F_{t}(u_i)}\right) =\frac{1}{\Delta + q}\left(\frac{1}{F_{t}(u_i) (1-\hat P(e_t))} - \frac{1}{F_{t}(u_i)}\right) = \frac{1}{\Delta + q} \cdot \frac{\hat P(e_t)}{1 - \hat P(e_t)} \cdot \frac{1}{F_t(u_i)},
    \end{align*}
    and identity \eqref{eq:diff_y_case2} follows by summing the above for all $u_i \in S_t \cap e_t$.

    Now $\E[Y_t \mid X_1,X_2,\dots,X_{t-1}] = Y_{t-1}$ follows by direct computation using \eqref{eq:diff_y_case1} and \eqref{eq:diff_y_case2}. Hence, all conditions of \Cref{def:martingales} are fulfilled, and indeed  $Y_0,\dots,Y_m$ forms a martingale w.r.t.\ $X_1,\dots,X_m$.
\end{proof}

\subsubsection{Bounding martingale parameters}
We recall that our goal is to prove that \Cref{eq:main_inequality}, i.e., that $Y = Y_m$ satisfies $Y \leq \frac{\Delta}{\Delta + q/2}$ with high probability (in $\Delta$). As $Y_0 = \frac{\deg(v)}{\Delta+q} \le \frac{\Delta}{\Delta + q}$ and trivially $Y\leq Y_0 + |Y-Y_0|$, it thus suffices to bound the difference between the first and last step of the martingale, $|Y - Y_0|$, as follows:
\begin{fact}[Sufficient Martingale Condition] \label{claim:bounding_y_diff}
    \Cref{eq:main_inequality} holds if
    \begin{equation} \label{eq:main_ineq_equiv_form}
        |Y - Y_0| \leq \frac{\Delta}{\Delta + q/2} - \frac{\Delta}{\Delta + q}.
    \end{equation}
\end{fact}

Our idea for proving that inequality \eqref{eq:main_ineq_equiv_form} holds often is to use specialized concentration inequalities for martingales, which do not require independence or explicit bounds on the positive correlation. Specifically, we will appeal to Freedman's inequality (see \Cref{thm:freedman_inequality}).
To this end, in the following two lemmas we upper bound this martingale's step size and observed variance.
\begin{lemma}[Step size]\label{lemma:bounding_step_size}
    For all times $t$, we have $|Y_{t} - Y_{t-1}| \leq A$, where $A:=  8/q$.
\end{lemma}
\begin{proof}
    By using the expressions for the difference $Y_t - Y_{t-1}$ from \cref{lemma:proof_y_is_martingale}, we obtain:
    \begin{align*}
        |Y_{t} - Y_{t-1}| & \leq \frac{1}{\Delta + q} \cdot \max \left \{ \frac{\hat P(e_t)}{1 - \hat P(e_t)}, 1 \right \} \cdot \sum_{u_i \in S_t \cap e_t} \frac{1}{F_{t}(u_i)} \\
        &\leq \frac{1}{\Delta + q} \cdot \sum_{u_i \in S_t \cap e_t} \frac{1}{q/(4\Delta)} \\
        & \leq \frac{1}{\Delta + q} \cdot \frac{2}{q/(4\Delta)} \\
        & \leq 8/q.
    \end{align*}
    For the second inequality, first notice that $\hat P(e_t) \leq P(e_t) \leq 1/2$ (by \cref{obs:Ftv_and_Pet_bounds}) which implies $\frac{\hat P(e_t)}{1 - \hat P(e_t)} \leq 1$. Also, we have $F_t(u_i) \geq q/(4\Delta)$ (by \cref{obs:Ftv_and_Pet_bounds}) at any point of time in the algorithm. For the third inequality we used the trivial fact that $|S_t \cap e_t| \leq 2$. \qedhere
\end{proof}

We next upper bound the observed variance $W_m$. While the following proof is computation-heavy, we note that all our following manipulations are straightforward, except for the insight that the hard lower bound $F(v)\geq q/(4\Delta)$ allows us to also bound $\sum_{e\in \delta(u)} \hat{P}(e)$.

\begin{lemma}[Observed Variance] \label{lemma:ub_variance_martingale}
    For the martingale $Y_t$ described above, we have:
    \begin{align}
    W_m := \sum_{t=1}^m \E[ (Y_t - Y_{t-1})^2 \mid X_1,\dots,X_{t-1} ] \leq \frac{128 \Delta \ln \Delta}{q^2}.
    \end{align}
\end{lemma}
\begin{proof}
    Using the expression for $Y_t - Y_{t-1}$ (see \cref{lemma:proof_y_is_martingale}), we first have:
    \begin{align*}
        &\E[ (Y_t - Y_{t-1})^2 \mid X_1, X_2, \ldots, X_{t-1}] \\
        = & \hat P(e_t) \cdot \left( \sum_{u_i \in S_t \cap e_t} \frac{1}{(\Delta + q)F_t(u_i)} \right)^2 + (1 - \hat P(e_t)) \cdot \left( \sum_{u_i \in S_t \cap e_t} \frac{\hat P(e_t)}{(\Delta + q)F_t(u_i) (1- \hat P(e_t))} \right)^2.
    \end{align*}
    Note that the $\hat{P}(e_t)$'s and $F_{t}(u_i)$'s depend on the variables $X_1, X_2, \ldots, X_{t-1}$ we are conditioning on, and that we will show the bound on the observed variance in \emph{any} execution of the algorithm.
    The above sums contain either one or two terms, as $|S_t \cap e_t| \leq 2$. By using the elementary inequality $(a + b)^2 \leq 2a^2 + 2b^2$, we obtain the following upper bound:
    \begin{align*}
        &\E[ (Y_t - Y_{t-1})^2 \mid X_1, X_2, \ldots, X_{t-1}] \\
        \leq & 2 \hat P(e_t) \cdot \sum_{u_i \in S_t \cap e_t} \left( \frac{1}{(\Delta + q)F_t(u_i)} \right)^2 + 2 (1 - \hat P(e_t)) \cdot \sum_{u_i \in S_t \cap e_t} \left( \frac{\hat P(e_t)}{(\Delta + q)F_t(u_i) (1-\hat P(e_t))} \right)^2.
    \end{align*}
    The above expression can be rewritten compactly by factoring out $\frac{2 \hat P(e_t)}{(\Delta + q)^2 \cdot (F_t(u_i))^2}$, which gives:
    \begin{align*}
        \E[ (Y_t - Y_{t-1})^2 \mid X_1, X_2, \ldots, X_{t-1}] \leq \sum_{u_i \in S_t \cap e_t} \frac{2 \hat P(e_t)}{(\Delta + q)^2 \cdot (F_t(u_i))^2} \left( 1 + \frac{\hat P(e_t)}{1-\hat P(e_t)} \right).
    \end{align*}
    Using \Cref{obs:Ftv_and_Pet_bounds}, we have $\hat P(e_t) \leq P(e_t) \leq 1/2$ and thus $1 + \frac{\hat P(e_t)}{1-\hat P(e_t)} \leq 2$, but also $F_t(u_i) \geq q/(4\Delta)$ for any $u_i \in S_t \cap e_t$.
    Additionally we note that $|S_t\cap e_t| \le 2$, and so we have:
    \begin{equation*}
        \E[ (Y_t - Y_{t-1})^2 \mid X_1, X_2, \ldots, X_{t-1}] \leq \sum_{u_i \in S_t \cap e_t} \frac{2 \hat P(e_t)}{(\Delta + q)^2} \cdot \frac{16\Delta^2}{q^2} \cdot 2 \le \frac{128 \hat P(e_t)}{q^2}.
    \end{equation*}

\begin{figure}[!t]

\captionsetup{justification=centering}
    \centering
    \includegraphics{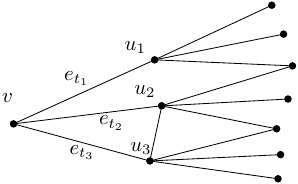}
    \caption{The configuration described in \cref{lemma:main_ineq}, i.e., the $2$-hop neighborhood of $v$. \\The (at most) $\Delta^2$ edges drawn correspond to the only non-trivial steps of the martingale.}
    \label{fig:two-neighbor}
\end{figure}
    
    By summing this inequality over all $t$ we will obtain an upper bound for $W_m$. Notice that an edge $e_t$ is thereby summed over on the right hand side only if it is incident to some vertex $u_i \in S_t$ at some time step~$t$ (see \Cref{fig:two-neighbor}). As $S_t \subseteq S_0 = N(v)$, we obtain the following upper bound, double counting edges $e_t$ that connect two distinct vertices of $S_t$ (for example edge $(u_2,u_3)$ in \Cref{fig:two-neighbor}):
    \begin{align} 
        W_m := \sum_{t = 1}^m \E[ (Y_t - Y_{t-1})^2 \mid X_1, X_2, \ldots, X_{t-1}] & \leq \sum_{u_i \in N(v)} \sum_{e_t \in \delta(u_i)} \frac{128 \hat P(e_t)}{q^2}. \label{eq:helper_eq_for_ub_for_martingale}
    \end{align}
    To upper bound the inner sum, fix a vertex $u_i \in N(v)$ and note that:
    \begin{equation*}
        \frac{q}{4\Delta} \leq F_m(u_i) = \prod_{e \in \delta(u_i)} (1-\hat P(e)) \leq \exp \left( -\sum_{e \in \delta(u_i)} \hat P(e) \right),
    \end{equation*}
    which implies that $\sum_{e \in \delta(u_i)} \hat P(e) \leq \ln \left( \frac{4\Delta}{q} \right) \leq \ln (\Delta)$ (using that $q\geq 8\sqrt{\Delta}\geq 8\geq 4$). By plugging the above into \eqref{eq:helper_eq_for_ub_for_martingale} and using the fact that $|N(v)| \leq \Delta$, we finally can conclude the proof of \cref{lemma:ub_variance_martingale}:
    \begin{align*}
        W_m := \sum_{t=1}^m \E[ (Y_t - Y_{t-1})^2 \mid X_1, X_2, \ldots, X_{t-1}] & \leq \frac{128 \Delta \ln \Delta}{q^2}. \qedhere
    \end{align*}
\end{proof}

\subsubsection{Conclusion of the analysis}
Having upper bounded both step size and variance of the martingale $Y_0,Y_1,\dots,Y_m=:Y$, we are now ready to leverage Freedman's inequality to prove that with high probability (in~$\Delta$), our desired upper bound on $|Y-Y_0|$ of inequality \eqref{eq:main_ineq_equiv_form} holds, and hence $F(v)\geq q/(3\Delta)$. 

\begin{lemma} \label{lemma:bound_on_Y}
    $\Pr[F(v)<q/(3\Delta)]\leq 2\Delta^{-3}.$
\end{lemma}
\begin{proof}
Let $\lambda := q/(3\Delta)$. By \Cref{claim:bounding_y_diff,lemma:main_ineq}, we have that 
\begin{align}\label{eqn:bring-Fv-forward}
    \Pr[F(v)<q/(3\Delta)] \leq \Pr\left[|Y - Y_0| \geq \frac{\Delta}{\Delta+q/2}-\frac{\Delta}{\Delta+q}\right] \leq  \Pr[|Y - Y_0| \geq \lambda].
\end{align}
Above, the second inequality used that 
\begin{equation*}
    \frac{q}{3\Delta} \leq \frac{\Delta}{\Delta + q/2} - \frac{\Delta}{\Delta + q},
\end{equation*}
which, by taking a common denominator, is seen to be equivalent to $3 \Delta q + q^2 \leq \Delta^2$, which follows directly from $q\leq \Delta / 4$ (see \eqref{eq:assumption}).

Next, to upper bound the RHS of \Cref{eqn:bring-Fv-forward}, we appeal to Freedman's inequality (\cref{thm:freedman_inequality}).
Letting $\sigma^2 := \frac{128 \Delta \ln \Delta}{q^2}$ and $A := 8/q$ (as in \cref{lemma:bounding_step_size,lemma:ub_variance_martingale}) and $\lambda$ as above, we have that:
\begin{align*} \label{eq:applying_friedman}
    \Pr[|Y - Y_0| \geq \lambda] & \leq 2 \exp\left( -\frac{\lambda^2}{2(\sigma^2 + A \lambda)}\right) \\
    & = 2 \exp\left( -\frac{(\frac{q}{3\Delta})^2}{2 \left( \frac{128 \Delta \ln \Delta}{q^2} + \frac{8}{q} \cdot \frac{q}{3\Delta} \right)}\right) \\
    & = 2 \exp \left( - \frac{200 \ln \Delta}{18 \left( 128/200 + 8/3 \cdot \Delta^{-1/2} \right)} \right) \\
    & \leq 2 \exp ( - 3 \ln \Delta) \\
    & = 2 \Delta^{-3},
\end{align*}
where for the second-to-last equality we used the definition of $q = \sqrt{200} \cdot \Delta^{3/4} \ln^{1/2}\Delta$. Combining the above with \Cref{eqn:bring-Fv-forward}, the lemma follows.
\end{proof}

To obtain the statement of \cref{thm:online_matching}, it remains to put the pieces together. With high probability in $\Delta$, the inequality from \cref{lemma:main_ineq} holds for both vertices incident at an edge $e$. Also, by a sequence of other lemmas we know that this inequality implies that $\hat P(e) = P(e)$ (with high probability in $\Delta$). Intuitively, this should guarantee the right marginal $1/(\Delta + O(q))$ of matching $e$. More formally, we now complete the proof of our main technical contribution:

\OnlineMatchingTheorem*
\begin{proof}
We prove that \cref{alg:edge-arrival} with $q$ as in this section is such an algorithm. Fix an edge $e_t = (u,v)$. The marginal probability that $e_t$ is matched by \cref{alg:edge-arrival} is given by $\Pr[X_t < \hat P(e_t)]$.
    We note that for the event $X_t<\hat P(e_t)$, we need that both $X_t<P(e_t)$ and $\hat P(e_t)=P(e_t)$ (else $\hat P(e_t)=0$ and so trivially $X_t\geq \hat P(e_t)$).
    We therefore have that \begin{align*}
    \Pr[X_t<\hat P(e_t)] & = \Pr\left[\left(X_t<P(e_t)\right) \wedge \left(\hat P(e_t) = {P}(e_t)\right)\right] \\
    & = \Pr[X_t<P(e_t)] -  \Pr\left[\left(X_t<P(e_t)\right) \wedge \left(\hat P(e_t) \neq {P}(e_t)\right)\right]
    \\
    & \geq \Pr[X_t<P(e_t)] - \Pr[\hat P(e_t)\neq P(e_t)] \\
    & \geq \Pr[X_t<P(e_t)] - (\Pr[F(v)< q/(3\Delta)]+\Pr[F(u)<q/(3\Delta)]) \\
    & \geq \frac{1}{\Delta+q} - 4\Delta^{-3} \\
    & \geq \frac{1}{\Delta + 4q}.
    \end{align*}
    Above, the second and third inequalities follow from \Cref{lemma:sufficient_cond_for_right_marginal} and \Cref{lemma:bound_on_Y} together with union bound over $u,v\in e_t$.
    The last inequality, on the other hand, 
    which is equivalent to $\frac{3q}{(\Delta+q)(\Delta+4q)}\geq 4\Delta^{-3}$, which clearly holds for large enough $\Delta$, and can be verified to hold for all $\Delta\geq 1$, using $8\sqrt{\Delta}\leq q\leq \Delta/4$, by \eqref{eq:assumption}.
\end{proof}

\paragraph{From Matching to Edge Coloring.}
We are finally ready to prove our main result, the existence of an online $\Delta(1+o(1))$-edge-coloring algorithm. To do so, it suffices to combine our online matching algorithm of \Cref{thm:online_matching} with the known reduction from online edge coloring to online matching given by \Cref{lemma:reduction2}, or alternatively, our strengthening of the latter in \cref{lemma:reduction_improved_appendix}. 
Combining these, we obtain the following quantitative result.

\begin{theorem} \label{thm:ec-body}
    There exists an online edge-coloring algorithm which, on $n$-vertex graphs of known maximum degree $\Delta = \omega(\log n)$,
    outputs with high probability (in $n$) a valid edge coloring using $\Delta + q'$ colors, where:
    \begin{equation*}
        q' = O(\Delta^{3/4}\log^{1/2}\Delta + \Delta^{2/3} \log^{1/3} n).
    \end{equation*}
\end{theorem}
\begin{proof}
     By \Cref{thm:online_matching}, there exists an online matching algorithm matching each edge with probability at least $1/(\alpha\Delta)$, for $\alpha := 1 + O(\Delta^{-1/4} \log^{1/2}\Delta)$. 
     But by \cref{lemma:reduction_improved_appendix}, such an online matching algorithm can be used to obtain an $(\alpha + O(\Delta^{-1/3}\log^{1/3}n)\cdot \Delta$-coloring algorithm (assuming $\Delta = \omega(\log n)$).
     Combining  \Cref{thm:online_matching} and \cref{lemma:reduction_improved_appendix}, we therefore obtain an online edge coloring algorithm that, w.h.p.~in $n$, on graphs of degree $\Delta=\omega(\log n)$, uses the claimed number of colors:
     \begin{align*}
         (\alpha + O(\Delta^{-1/3}\log^{1/3}n))\cdot \Delta & = \Delta + O(\Delta^{3/4}\log^{1/2}\Delta + \Delta^{2/3} \log^{1/3} n). \qedhere 
     \end{align*}
\end{proof}

%% file: Sections/applications.tex
\section{Online List Edge Coloring and Local Edge Coloring}
\label{sec:matching_peeling}

In this section we present a number of applications that follow from our \Cref{alg:edge-arrival}. In particular, we show generalizations to
\emph{online list edge coloring} (\cref{thm:list_edge_coloring}, in \Cref{sec:list-coloring}) and  \emph{online local edge coloring} (\cref{thm:local_edge_coloring_strenghtened}, in \Cref{sec:local-coloring}).

In the \underline{online \emph{list} edge coloring problem}, every online arriving edge $e$ presents a list $L(e)\subseteq \mathbb{N}$ of colors that can be used for coloring $e$. Unlike in the classical setting, this list $L(e)$ is not necessarily of the form $\{1,\dots,\Delta + q\}$ for some $q\geq 1$, but can instead be arbitrary. The objective of an algorithm is to provide a valid coloring of all edges with each edge $e$ assigned color in its list, $c(e)\in L(e)$, and no vertex having two incident edges assigned the same color. Our result for this problem is given by:
\begin{restatable}[Online List Edge Coloring]{theorem}{ListColoringTheorem} \label{thm:list_edge_coloring}
    There exists an online list edge-coloring algorithm which, on $n$-vertex graphs of known maximum degree $\Delta$,
    outputs with high probability (in $n$) a valid list-edge coloring, provided all lists $L(e)$ satisfy $|L(e)| \geq \Delta + q$, for $$q := 10^{24} \ln n + 10^4 \cdot \left( \Delta^{3/4} \ln^{1/2} \Delta + \Delta^{2/3} \ln^{1/3} n \right).$$
    In particular, if $\Delta = \omega(\log n)$, then lists of size $\Delta+q$ with $q = o(\Delta)$ suffice.
\end{restatable}

In the \underline{online \emph{local} edge coloring problem}, the setting is the same as in the classical version. However, we now aim to color each edge $e = (u,v)$ using a color of index $c(e)$ that is not much larger than the max degree of its endpoints, $d_{\max}(e):=\max\{\deg(u), \deg(v)\}$.
Here we prove the following result:
\begin{restatable}[Online Local Edge Coloring]{theorem}{LocalColoringTheorem}
 \label{thm:local_edge_coloring_strenghtened}
    There exists an online edge-coloring algorithm which, on $n$-vertex graphs with \emph{a priori} known degree sequence $\{\deg(v)\mid v\in V\}$,\footnote{This theorem also holds if $\deg(v)$ are only \emph{upper bounds} for the true degrees, in which case the guarantees of the theorem will be with respect to these upper bounds.} computes, with high probability (in $n$), an edge coloring $c : E \to \mathbb{N}$ that colors each edge $e$ using a color $c(e)$ which satisfies:
    \begin{align*}
        c(e) \leq d_\text{max}(e) + 2 \cdot 10^{24} \ln n + 10^5 \cdot \left( d_{\text{max}}^{3/4}(e) \ln^{1/2} d_{\text{max}}(e) + d_{\text{max}}^{2/3}(e) \ln^{1/3} n \right).
    \end{align*}
    In particular, if $d_{\text{max}}(e) = \omega(\log n)$, then $c(e) \leq d_{\max}(e)\cdot (1+ o(1))$.
\end{restatable}

To obtain both generalizations of \Cref{thm:edge-coloring} above, we rely on an online coloring subroutine, \cref{alg:reducing_max_degree_generic}, which we provide in the following section.
This algorithm colors a (potentially strict) subset of the edges of the graph, and assigns each colored edge $e$ a color from some individual (small) list of colors $\ell(e)$, revealed when $e$ arrives. 
Under the mild technical condition (which holds with high probability for our invocations of this algorithm in later sections) that $|\ell(e)|\approx \lambda$ for some appropriate parameter $\lambda$, the subroutine reduces the maximum degree $\Delta$ of the remaining uncolored subgraph by $\approx \lambda$. 
In other words, for all high-degree vertices in $U$ (those with degree $\approx \Delta$), the algorithm colors $\approx \lambda$ of their incident edges. This subroutine is heavily inspired by \cite[Algorithm~4]{cohen2019tight}, which uses (instead of our \Cref{alg:edge-arrival}) an older online matching rounding algorithm, called MARKING \cite{cohen2018randomized, wajc2020matching}, for sampling matchings. 
Unlike our subroutine which works under edge arrivals, the previous \cite[Algorithm~4]{cohen2019tight} is restricted to one-sided vertex arrivals, because MARKING is also restricted to this regime.
Moreover, our algorithm's more general ability to color from lists $\ell(e)$ for each edge $e$ allows us to obtain both the list edge coloring result and local edge coloring results, which could not be obtained from \cite[Algorithm~4]{cohen2019tight}.

In Section \ref{subsec:generic_coloring}, we present \cref{alg:coloring_generic}, which applies \cref{alg:reducing_max_degree_generic} successively on monotonically decreasing subgraphs of $G$ to color edges. 
This coloring algorithm is very flexible, in the sense that every arriving edge $e$ presents a personalized list $L(e)$ of colors that can be used to color it. 
This algorithm forms the underpinning of both theorems above, proven in later sections.

\subsection{Subroutine to Reduce the Maximum Degree}

In this section we provide an algorithmic subroutine for reducing the uncolored maximum degree. This will be used for our list and local edge coloring results, as well as our improvement on the reduction from edge coloring to matching of \Cref{lemma:reduction2}.

The psueocode of this algorithmic suroutine is given in \Cref{alg:reducing_max_degree_generic}, and it works as follows.
Its input is a graph $U$ with $n$ vertices and maximum degree upper bounded by some $\Delta(U)$ arriving online edge-by-edge. (We think of $U$ as the \emph{uncolored} subgraph of an initial graph $G$ on the same vertex set.) Let $\lambda := \Delta(U)^{2/3} \ln^{1/3} n$. We call a vertex $v \in V$ \emph{dense} if $\deg_U(v) \geq \Delta(U) - \lambda$.
For any arriving edge $e$, the algorithm receives a list $\ell(e)$ of available colors. If $e$ is incident to a dense vertex $v$, the list $\ell(e)$ is guaranteed to have size $|\ell(e)|\approx \lambda$. 
For each new color revealed, we run a copy of \Cref{alg:edge-arrival}, guaranteeing that each edge is matched (and hence) colored by each of the colors $c\in \ell(e)$ with probability $\approx \frac{1}{\Delta}$.

\begin{figure}[h!]
\begin{center}
\begin{minipage}{0.95\textwidth}
\begin{mdframed}[hidealllines=true, backgroundcolor=gray!15]

\begin{algorithm}[Coloring Algorithm Reducing Maximum Degree] \ \\

\item \textbf{Input:} Graph $U$ with $n$ vertices, arriving edge-by-edge. Each edge $e$ arrives with list of colors $\ell(e)$. 

\textbf{Promise:} A value $\Delta(U)$ upper bounding the maximum degree of $U$ is given. 

\textbf{Promise:} For any edge $e$ incident to a dense vertex, we have  $\lambda\leq |\ell(e)| \leq \lambda + 10 \sqrt{\lambda \ln n}$, for $\lambda := \Delta(U)^{2/3} \ln^{1/3} n$.
\item \textbf{Output}: Coloring of a subset of edges of $U$.

\item When edge $e$ together with list $\ell(e)$ arrives:
\item Iterate through $c \in \ell(e)$:
\begin{itemize}
\item If $c$ was never seen before, launch a new instance of \Cref{alg:edge-arrival} corresponding to $c$ on $U$, using $\Delta(U)$ as an upper bound for the maximum degree of $U$.
\item Input the edge $e$ to the instance of \Cref{alg:edge-arrival} corresponding to $c$.
\item If $e$ was matched by \cref{alg:edge-arrival} and was not already colored, color $e$ with color $c$.
\end{itemize}
\label{alg:reducing_max_degree_generic}
\end{algorithm}
\end{mdframed}
\end{minipage}
\end{center}
\end{figure}

As we now show, if the maximum degree of $U$ is sufficiently large(r than $\ln n$), the above algorithm guarantees with high probability that dense vertices have $\approx \lambda$ edges colored by this algorithm, and therefore the maximum degree of the remaining uncolored graph decreases by roughly $\lambda$ as well.

\begin{theorem} \label{thm:reducing_max_degree_generic}
    Every edge $e$ colored by \Cref{alg:reducing_max_degree_generic} 
    is assigned a color $c \in \ell(e)$. 
    Moreover, after running \Cref{alg:reducing_max_degree_generic} on $U$ with $\Delta(U)\geq 10^{24}\ln n$, and using \Cref{alg:edge-arrival} with $q := 10^2 \cdot \Delta(U)^{3/4} \ln^{1/2} \Delta(U)$, with the necessary promises, the maximum degree of the remaining uncolored subgraph $U'\subseteq U$ is upper bounded by $\Delta(U') := \Delta(U) - \lambda + 2 \lambda \cdot \frac{q + \lambda}{\Delta + q} + 6 \sqrt{\lambda \ln n} \leq \Delta(U)-0.9\lambda$ with probability at least $1 - \frac{1}{n^{15}}$.
\end{theorem}
Before proving the above theorem, it will be convenient to present some useful inequalities in a separate lemma, whose technical proof (which are easy to verify hold asymptotically, but also hold for all choices of $n\geq 2$ due to the large constants chosen) is deferred to the end of the section:
\begin{fact} \label{lemma:aux_lemma}
    Suppose $\Delta(U) \geq 10^{24} \ln n$ and define the following parameters: 
    \begin{align*} 
    d & := \Delta(U),\\
    q & := 10^2 \cdot d^{3/4} \ln^{1/2} d \\
    \lambda & := d^{2/3} \ln^{1/3} n \\
    \mu & := (d - \lambda) \cdot \frac{\lambda}{d + q} \left( 1 - \frac{\lambda}{d + q} \right).
    \end{align*} 
    Then, these parameters satisfy the following inequalities:
    \begin{align}
        &\mu \geq \lambda - 2 \lambda \cdot \frac{q + \lambda}{d + q} \label{ineq:aux_1} \\
        & 0.1 \lambda \geq 2 \lambda \cdot \frac{q + \lambda}{d + q} + 6\sqrt{\lambda \ln n} \label{ineq:aux_2} \\
        & \mu \geq 30\ln n. \label{ineq:aux_3}
    \end{align}
\end{fact}

\begin{proof}[Proof of \cref{thm:reducing_max_degree_generic}]
   That each edge $e$ is assigned a color (if any) from $\ell(e)$ is immediate from the algorithm's description. The ``meat'' of the proof is therefore in proving that removal of the edges colored yields a subgraph of maximum degree as low as claimed in the theorem statement.
   
   Consider a dense vertex $v$ of $U$, i.e., with degree $\deg_U(v) \geq d - \lambda$. (If no such vertex exists, the statement to prove follows trivially.) By the proof of \cref{thm:online_matching}, each edge $e$ incident to $v$ is picked by color $c \in \ell(e)$ with probability
    at least $\frac{1}{d + 4q'}\geq \frac{1}{d+q}$, for $q'=\sqrt{200}d^{3/4}\ln^{1/2}d$, where $4q'\leq q$. Let $X_e$ be the indicator variable for the event that $e$ is colored by \cref{alg:reducing_max_degree_generic} with any color $c \in \ell(e)$. As $|\ell(e)| \geq \lambda$ for dense vertices, we have by the first two terms of the Taylor expansion of $\exp(-\lambda/(d+q))$ for $\lambda/(d+q)\leq 1$ (as in our case):
    \begin{equation*}
        \Pr[X_e = 1] \geq 1 - \left(1 - \frac{1}{d + q}\right)^\lambda \geq 1-\exp\left(-\frac{\lambda}{d+q}\right)\geq \frac{\lambda}{d+q} - \left(\frac{\lambda}{d+q}\right)^2.
    \end{equation*}
    Let $X := \sum_{e \in \delta(v)} X_e$ be the number of edges incident to $v$ that are colored by \cref{alg:reducing_max_degree_generic}.
    We wish to argue that $X$ is not much less than the following lower bound $E[X] \geq \mu := (d - \lambda) \cdot \frac{\lambda}{d + q} \left( 1 - \frac{\lambda}{d + q} \right)$ on its expectation. (This lower bound follows by our lower bound on $\Pr[X_e=1]$, linearity of expectation, and the lower bound on the number of edges of dense vertices.)
    
    We now argue that with high probability in $n$, the expectation $\E[X]$ does not fall short of this lower bound of $\mu$. This would follow easily by standard Chernoff bounds if the $X_e$ were independent, though they clearly are not. However, we can interpret these variables as indicator variables of a balls and bins process, which are \emph{negatively associated (NA)}, and hence admit the same Chernoff bounds as independent random variables \cite{dubhashi1996balls,wajc2017negative}.
    In more detail, we have a ball for each color $c$, and each ball (color) $c$ falls into a (single) bin corresponding to either an edge $e$ of $v$ or a dummy bin, depending on whether or not $e$ is matched by the $c^{th}$ copy of \Cref{alg:edge-arrival}, which happens independently (but not i.i.d) for different $c$.
    The values $X_e$ are therefore indicators for whether bin $e$ is non-empty, and these random variables are NA \cite[Corollary 2.4.6]{wajc2020matching}.
    Consequently, the sum $X$ of the NA (but not independent) random variables $X_e$, which satisfies $\E[X]\geq \mu$, also satisfies the following standard Chernoff bound for any $\varepsilon<1$:
    \begin{equation}\label{eq:helper_chernoff_aux}
        \Pr \left[ X \leq (1 - \epsilon) \cdot \mu \right] \leq \exp \left( - \frac{\varepsilon^2 \cdot \mu}{2} \right).
    \end{equation}
    Fix $\varepsilon:= \sqrt{\frac{30 \ln n}{\mu}}$. By \Cref{ineq:aux_3} from \Cref{lemma:aux_lemma}, have that $\varepsilon\in [0,1]$.
    On the other hand, as we shall soon see, we also have that
    \begin{equation} \label{eq:helper_thm:reducing_max_degree_generic}
        (1 - \varepsilon)\cdot \mu \geq \lambda - 2 \lambda \cdot \frac{q + \lambda}{\Delta + q} - 6 \sqrt{\lambda \ln n},
    \end{equation}
    To see this, first note that, as $\mu = \lambda \cdot \frac{d-\lambda}{d+q} \cdot \left( 1 - \frac{\lambda}{d + q} \right) \leq \lambda$, we have that $\varepsilon \cdot \mu \leq \sqrt{30 \lambda \ln n} \leq 6 \sqrt{\lambda \ln n}$. Furthermore, $\mu \geq \lambda - 2 \lambda \cdot \frac{q + \lambda}{\Delta + q}$, by \eqref{ineq:aux_1} from \cref{lemma:aux_lemma}. We obtain \eqref{eq:helper_thm:reducing_max_degree_generic} by summing up these last two inequalities. 
\color{black} 
    Consequently, combining \eqref{eq:helper_chernoff_aux} and \eqref{eq:helper_thm:reducing_max_degree_generic} and using our choice of $\varepsilon$, we have that 
    \begin{equation*}
        \Pr \left [ X \leq \lambda - 2 \lambda \cdot \frac{q + \lambda}{\Delta + q} - 6 \sqrt{\lambda \ln n} \right] \leq \Pr [X \leq (1 - \varepsilon)\cdot \mu] \leq \frac{1}{n^{15}}.
    \end{equation*}
    The claimed upper bound on the maximum degree of $U'$ then follows, since this upper bound holds for all non-dense vertices in $U$ (which is a super graph of $U'$), together with union bound over all dense vertices $v$, which all have degree at most $\Delta(U)$ in $U$, and so, with high probability in $n$, in the new graph $U'$ all vertices have degree at most $\Delta(U)-\left(\lambda - 2 \lambda \cdot \frac{q + \lambda}{\Delta + q} - 6 \sqrt{\lambda \ln n}\right)\leq \Delta(U)-0.9\lambda$, where the last inequality follows from \eqref{ineq:aux_2} from \Cref{lemma:aux_lemma}.
\end{proof}

The inequalities from \cref{lemma:aux_lemma} are obtained by direct computation as follows:
\begin{proof}[Proof of \cref{lemma:aux_lemma}]
    \eqref{ineq:aux_1} follows easily:
    \begin{align*}
        \mu & = (d - \lambda) \cdot \frac{\lambda}{d + q} \left( 1 - \frac{\lambda}{d + q} \right) = \lambda \cdot \left( 1 - \frac{q + \lambda}{d + q} \right) \left( 1 - \frac{\lambda}{d + q} \right) \\
        &\geq \lambda - \lambda \cdot \frac{q + \lambda}{d + q} - \lambda \cdot \frac{\lambda}{d + q} \geq \lambda - 2 \lambda \cdot \frac{q + \lambda}{d + q}.
    \end{align*}
    To get \eqref{ineq:aux_2} it suffices to show:
    \begin{align}
        &2 \lambda \cdot \frac{q + \lambda}{d + q} \leq 0.05 \lambda \label{ineq:aux_4} \\
        &6 \sqrt{\lambda \ln n} \leq 0.05 \lambda. \label{ineq:aux_5}
    \end{align}
    For \eqref{ineq:aux_4}, notice that:
    \begin{align*}
        2 \cdot \frac{q + \lambda}{d + q} \leq 2 \cdot \frac{q + \lambda}{d} =
        \frac{2 \cdot 10^2 \cdot d^{3/4} \ln^{1/2} d}{d} + \frac{2 \cdot d^{2/3} \ln^{1/3} n}{d}.
    \end{align*}
    The first term can be upper bounded by $0.025$ for any $n \geq 2$ given $d\geq 10^{24}\ln n$. The second term can be upper bounded by $\frac{2 \ln^{1/3} n}{d^{1/3}}\leq 2/10^8 \leq 0.025$, since $d \geq 10^{24} \ln n$. Summing the upper bounds for the two terms gives \eqref{ineq:aux_4}.
    Inequality \eqref{ineq:aux_5} follows from $\sqrt{\lambda/\ln n} \geq 10^8 \geq 120$, again since $d \geq 10^{24} \ln n$ and so $\lambda = d^{2/3}\ln^{1/3}n \geq 10^{16}\ln n$. 
    
    It remains to prove \eqref{ineq:aux_3}. By \eqref{ineq:aux_1} and \eqref{ineq:aux_2}, we have that $\mu \geq 0.9 \lambda$. And indeed, again using that $d \geq 10^{24} \ln n$ and so $\lambda\geq 10^{16}\ln n$, we obtain the desired inequality, as $\mu \geq 0.9\cdot 10^{16}\ln n\geq 30 \ln n$.     
\end{proof}
\color{black}

We now turn to using this subroutine to completely color a given graph revealed online.

\subsection{Generic Coloring Algorithm} \label{subsec:generic_coloring}

\paragraph{Strategy.} For subsequent applications, we consider graphs $G$ with a maximum degree known to be at most $\Delta$, with edges arriving online one by one. Moreover, every arriving edge $e$ provides a list of available colors $L(e)$. We will repeatedly apply \cref{alg:reducing_max_degree_generic}, consuming different sublists $\ell(e) \subseteq L(e)$ of colors per phase, to reduce the maximum degree of the currently uncolored graph $U \subseteq G$, by coloring subsets of the edges of $U$. The following definition introduces the relevant parameters and is used to define \cref{alg:coloring_generic}:
\begin{definition}[Degree Sequence] \label{def:deg_seq}
    Let $d_0 := \Delta$ be (an upper bound on) the maximum degree of a graph $G$ with $n$ vertices. For $i\geq 0$, we define the following:
    \begin{align*}
        &\lambda_i := d_i^{2/3} \ln^{1/3} n, \\
        &q_i := 10^2 \cdot d_i^{3/4} \ln^{1/2} d_i, \\
        &d_{i+1} := d_i - \lambda_i + 2 \lambda_i \cdot \frac{q_i + \lambda_i}{d_i + q_i} + 6\sqrt{\lambda_i \ln n} \leq d_i.
    \end{align*}
    Let $f$ be the minimal value for which $d_f<10^{24}\ln n$. (Such an $f$ exists, since $d_i \geq 10^{24} \ln n$, then by \cref{lemma:aux_lemma}, $d_{i+1} \leq d_i - 0.9\lambda_i \leq d_i - 1$.)
    We call the parameters $D(d_0) := \{d_i : 0 \leq i \leq f + 1\}$ the \emph{degree sequence of} $d_0$.
\end{definition}

\begin{figure}[h!]
\begin{center}
\begin{minipage}{0.95\textwidth}
\begin{mdframed}[hidealllines=true, backgroundcolor=gray!15]

\begin{algorithm}[Generic Coloring Algorithm] \ \\

\item \textbf{Input:} Graph $G$ with $n$ vertices, arriving edge-by-edge together with lists $L(e)$ of available colors. We use the notations introduced in \cref{def:deg_seq}, and denote by  
$\mathcal{C} := \cup_{e} L(e)$ the set of all colors.

\item \textbf{Output}: Coloring of a subset of edges of $G$.

\begin{itemize}

\item Let $C_0,\dots,C_{f+1}$ be a partitioning of $\mathcal{C}$ which is computed online. 

\item Set $U_0\gets G$.

\item Iterate through phases $i \in \{0,\dots,f\}$:

\begin{itemize}
\vspace{-2mm}
\item Apply \cref{alg:reducing_max_degree_generic} (online) on the currently uncolored graph $U_i$.
\item For any edge $e$ incident to a \emph{dense} vertex in $U_i$ (i.e., having degree $\geq d_i - \lambda_i$), use the sublist $\ell^i(e) := L(e) \cap C_i$ to input online to \cref{alg:reducing_max_degree_generic}.
\item Set $U_{i+1}\gets U_i \setminus \{\textrm{edges colored by \Cref{alg:reducing_max_degree_generic} in phase } $i$\}$

\end{itemize}

\item Try to color the final uncolored graph $U_{f+1}$ using Greedy with the remaining lists of available colors $\ell^{f+1}(e) := L(e) \cap C_{f+1}$ for each edge $e$.
\end{itemize}

\label{alg:coloring_generic}
\end{algorithm}
\end{mdframed}
\end{minipage}
\end{center}
\end{figure}

Let $\mathcal{C} := \cup_{e \in G} L(e)$ be the set of all colors. Note that this set is \emph{a priori} unknown to the online algorithm and is only revealed indirectly through the lists $L(e)$ of the arriving edges. \cref{alg:coloring_generic} partitions this set into $f+2$ subsets of colors $C_0,\dots,C_{f+1}$ online. More concretely, whenever the arriving list $L(e)$ of an edge $e$ contains a color $c \in L(e)$ which is seen for the first time, it decides online to which of the sets $C_0,\dots,C_{f+1}$ color $c$ will be assigned. \emph{We insist on the fact that this choice can be made arbitrarily depending on the application. However, there is a restriction which we discuss in the next paragraph.}

The \emph{$i$-th phase} is the iteration of \cref{alg:coloring_generic} on the uncolored subgraph $U_i$. The partitioning of $\mathcal{C}$ into $C_0,\dots,C_{f+1}$ defines the colors to be used in each phase. Thereby, $\ell^i(e) := L(e) \cap C^i$ is the chosen sublist of $L(e)$ which is inputted online to \cref{alg:reducing_max_degree_generic} for edge $e$ during the execution of the $i$-th phase. Let $L^i(e) := L(e) \setminus \cup_{j = 0}^{i-1} \ \ell^j(e)$ be the set of colors which are still available to use for $e$ during the $i$-th phase. For all edges $e$, the sublists $\ell^i(e) \subseteq L^i(e)$ must have the size required to apply \cref{thm:reducing_max_degree_generic}, that is, $\lambda_i \leq |\ell^i(e)| \leq \lambda_i + 10 \sqrt{\lambda_i \ln n}$. This motivates the following:
\begin{definition}[Admissible Partitioning of Colors] \label{def:admissible_partitioning}
    Given a fixed input of \cref{alg:coloring_generic}, let $C_0,\dots,C_{f+1}$ be a partitioning of the set of colors $\mathcal{C} := \cup_{e \in G} L(e)$ which is computed online. 
    We say that the partitioning $C_0,\dots,C_{f+1}$ is \emph{admissible} if for any edge $e$ and phase $i \in \{0,\dots,f\}$ in which $e$ is incident to a \emph{dense} vertex in $U_i$ (i.e., having degree $\geq d_i - \lambda_i$),
    we have that $\ell^i(e) := L(e) \cap C_i$ satisfies:
    \begin{equation*}
        \lambda_i \leq |\ell^i(e)|\leq \lambda_i + 10 \sqrt{\lambda_i \ln n}.
    \end{equation*}
\end{definition}

By successively applying \cref{alg:reducing_max_degree_generic} (starting with the initial graph $U_0 := G$), as done in \cref{alg:coloring_generic}, and \emph{using an admissible (online) partitioning of $\mathcal{C}$ into $C_0,\dots,C_{f+1}$ as defined in \cref{def:admissible_partitioning}}, one obtains a sequence of subgraphs $U_0 \supseteq U_1 \supseteq \cdots \supseteq U_{f+1}$ of $G$, such that by successive applications of \cref{thm:reducing_max_degree_generic} we obtain the following:
\begin{lemma} \label{lemma:deg_seq_property}
    With high probability in $n$, the graphs $U_i$ computed online by \cref{alg:coloring_generic} (consisting of yet uncolored edges) have their maximum degree bounded by $d_i$. Moreover, for $i \leq f$ we have $d_i \geq 10^{24} \ln n$.
\end{lemma}
\begin{proof}
    Let $\text{good}_i$ be the event that $U_i$ has its degree upper bounded by $d_i$, where $i \in \{0,\dots,f\}$. Clearly, $\text{good}_0$ holds by the fact that $d_0$ is an upper bound on the maximum graph of the initial graph $G$. Assuming $\text{good}_i$ holds, consider the $i$-th phase of \cref{alg:coloring_generic}. Then, by \cref{thm:reducing_max_degree_generic}, the probability that $\text{good}_{i+1}$ holds is at least $1 - \frac{1}{n^{10}}$. Hence:
    \begin{equation} \label{eq:lower_bound_goodness}
        \Pr[\text{good}_{i+1} \mid \text{good}_{i}] \geq 1 - \frac{1}{n^{10}}.
    \end{equation}
    By induction it easily follows that $\Pr[\text{good}_{i}] \geq 1 - \frac{i}{n^{10}}$ for any $i \in \{0,\dots,f\}$. For $i = 0$ this is true with with probability $1$, and for the induction step we use \eqref{eq:lower_bound_goodness} together with the induction hypothesis to obtain:
    \begin{equation*}
         \Pr[\text{good}_{i+1}] \geq \Pr[\text{good}_{i+1} \mid \text{good}_{i}] \cdot \Pr[\text{good}_{i}] \geq \left( 1 - \frac{1}{n^{10}} \right) \cdot \left( 1 - \frac{i}{n^{10}} \right) \geq 1 - \frac{i+1}{n^{10}}.
    \end{equation*}
    In particular, for the complementary events we have $\Pr[\overline{\text{good}_{i}}] \leq \frac{1}{n^9}$ for any $i$, and the lemma follows by union bound over the at most $n$ possible values of $i$:
    \begin{align*}
        \Pr \left[ \cup_{i \in \{0,\dots f\}} \ \overline{{\text{good}_{i}}} \right] & \leq n \cdot \frac{1}{n^9} \leq \frac{1}{n^8}. \qedhere
    \end{align*}
\end{proof}

\paragraph{Analysis.} We have designed a generic coloring algorithm which, with high probability, successively reduces the maximum degree of the uncolored subgraphs $U_i$ until their maximum degree drops below $O(\log n)$. It remains to argue how to ensure the following properties required implicitly by \cref{alg:coloring_generic}:
\begin{itemize}
    \item The lists $L^i(e)$ of remaining colors need to have sufficiently large size to allow extracting the sublists $\ell^i(e) \subseteq L^i(e)$, such that $|\ell^i(e)| \geq \lambda_i$ as required by the application of \cref{alg:reducing_max_degree_generic} inside \cref{alg:coloring_generic}.
    \item In particular, to color $U_{f+1}$ successfully using Greedy, one needs to ensure that $L^{f+1}(e) \geq 2 d_{f+1}$.
\end{itemize}
To obtain these guarantees, we maintain by induction, throughout all phases $i$, the property $|L^i(e)| \geq d_i + a_i$ for edges connected to dense vertices, where $a_i$ is a large enough slack. We define these slacks precisely:
\begin{definition}[Slack Sequence] \label{def:slck_seq}
    Let $d_0 \geq 1$ and introduce the degree sequence $D(d_0)$ of $d_0$, and the parameters $f, \lambda_i, q_i$ as in \cref{def:deg_seq}. We define the \emph{slack sequence} $Sl(d_0) := \{a_i : 0 \leq i \leq f+1\}$ where:
    \begin{align*}
        &a_{f+1} := 10^{24} \ln n > d_{f+1} \\
        &a_i := a_{i+1} + 2 \lambda_i \cdot \frac{q_i + \lambda_i}{d_i + q_i} + 16\sqrt{\lambda_i \ln n}.
    \end{align*}
\end{definition}
As anticipated, we prove by induction that these slacks fulfill the required guarantees:
\begin{lemma} \label{lemma:slck_seq_property}
    Assume that $|L(e)| \geq 2 \cdot 10^{24} \ln n$ for all edges $e$. Furthermore, assume that, before the execution of the $i$-th phase of \cref{alg:coloring_generic}, for any edge $e$ connected to a dense vertex, $|L^i(e)| \geq d_i + a_i$. Then, after the execution of the $i$-th phase, one has $|L^{i+1}(e)| = |L^i(e) \setminus \ell^i(e)| \geq d_{i+1} + a_{i+1}$. Furthermore, executing Greedy on $U_{f+1}$ in \cref{alg:coloring_generic} is possible, as $|L^{f+1}(e)| \geq 2d_{f+1}$ for all edges $e \in U_{f+1}$.
\end{lemma}
\begin{proof}
    Using $|\ell^i(e)| \leq \lambda_i + 10 \sqrt{\lambda_i \ln n}$, we obtain:
    \begin{align*}
        |L^{i+1}(e)| = |L^i(e)| - |\ell^i(e)| \geq d_i + a_i - |\ell^i(e)| \geq d_i + a_i - (\lambda_i + 10 \sqrt{\lambda_i \ln n}) = d_{i+1} + a_{i+1},
    \end{align*}
    where the last inequality follows by the definitions of $a_i$ and $d_{i+1}$ (see \cref{def:deg_seq} and \cref{def:slck_seq}).
    As for the property $|L^{f+1}(e)| \geq 2d_{f+1}$, if $e$ was ever connected to a dense vertex in some phase $i$, then the property follows because $|L^{f+1}(e)| \geq d_{f+1} + a_{f+1} \geq 2d_{f+1}$, where the last inequality is due to $a_{f+1} > d_{f+1}$ (see \cref{def:slck_seq}). If on the contrary $e$ was never connected to a dense vertex, then $L^{f+1}(e) = L(e)$ (i.e.\ the list of available colors never changed during the execution of \cref{alg:coloring_generic}) and so we have $|L^{f+1}(e)| = |L(e)| \geq 2 \cdot 10^{24} \ln n > 2d_{f+1}$.
\end{proof}

To finish the analysis, it remains to upper bound the slacks $a_i$ in closed form, as opposed to the (convenient but) recursive definition from \cref{def:slck_seq}. In particular, the upper bounds on $a_i$ indicate how large the lists $L^i(e)$ need to be to make \cref{alg:coloring_generic} succeed. It is clear that:
\begin{align}
    a_i &= 10^{24} \ln n + \sum_{j = i}^f \left(2 \lambda_j \cdot \frac{q_j + \lambda_j}{d_j + q_j} + 16\sqrt{\lambda_j \ln n} \right) \nonumber \\
    & \leq 10^{24} \ln n +  \sum_{j = i}^f \left(2 \lambda_j \cdot \frac{q_j}{d_j} + 2 \lambda_j \cdot \frac{\lambda_j}{d_j} + 16\sqrt{\lambda_j \ln n} \right) \nonumber \\
    & \leq 10^{24} \ln n +  (f - i + 1) \cdot \left(2 \lambda_i \cdot \frac{q_i}{d_i} + 2 \lambda_i \cdot \frac{\lambda_i}{d_i} + 16\sqrt{\lambda_i \ln n} \right), \label{eq:aux_slack_upper_bound}
\end{align}
where the last inequality follows because all three terms inside the parentheses are non-increasing in $j$. By the above manipulations, it remains to upper bound the quantity $f - i + 1$, where $f$ is the number of phases. For this purpose, the following lemma is helpful:
\begin{lemma} \label{lemma:aux_local_result}
    Let $a > 0$ be any number.
    Consider a sequence $(x_k)_{k \geq 0}$ of non-negative integer numbers, such that $x_0 \geq 1$ and:
    \begin{equation}
        x_{k+1} := 
        \begin{cases}
            x_k - \ceil{a \cdot x_k^{2/3}} & \text{if $x_k \geq 1$,} \\
            0 & \text{if $x_k < 1$.}
        \end{cases}
    \end{equation}
    Then, for (integer) $k \geq \frac{3 \sqrt[3]{x_0}}{a}$, we have $x_k \leq 1$.
\end{lemma}
\begin{proof}
    We prove the statement by strong induction on the starting value $x_0 \geq 1$ of the sequence.
    If $x^{1/3}_0 < a$, then $a x_0^{2/3} > x_k$, and in particular, $x_1 = x_0 - \lceil a\cdot x_0^{2/3}\rceil < 0$, so the statement holds.

    Now assume instead that $x^{1/3}_0 \ge a$, and that the statement holds for any integer $x'_0 < x_0$.
     We apply the induction hypothesis on the sequence starting with $x'_0 := x_1 = x_0 - \ceil{a \cdot x_0^{2/3}}$ (trivially $x_1 < x_0$). We can assume $x_1 \ge 1$, else the statement already holds.
     The induction hypothesis gives us that $x_{k+1} \leq 1$ for $k \geq \frac{3 \sqrt[3]{x_1}}{a}$. It thus suffices to prove:
     \begin{align*}
     \frac{3\sqrt[3]{x_1}}{a} +1 &\le \frac{3\sqrt[3]{x_0}}{a},
     \\
     \intertext{which we rearrange into:}
     \sqrt[3]{x_1} &\le \sqrt[3]{x_0} - \frac{a}{3}.
     \\
     \intertext{Now taking the third power of both sides (which are indeed both positive):}
     x_0-\lceil a\cdot x_0^{2/3} \rceil &\le x_0 - 
     a x^{2/3}_0
     +\frac{a^2 x^{1/3}_0}{3}
     -\frac{a^3}{9}.
     \end{align*}
     The above inequality follows from the fact that
     $\lceil a x_0^{2/3}\rceil \ge a x_0^{2/3}$, and $a^2x_0^{1/3}/3 \ge a^3/3 > a^3/9$ (since we assumed $x_0^{1/3} \ge a$). Thus the induction proof is concluded and the lemma proven.
\end{proof}
We are now ready to upper bound $f - i + 1$.
\begin{lemma} \label{lemma:bounding_num_phases}
    For every $i \in \{0,\dots,f\}$, one has $f - i + 1 \leq 7 \sqrt[3]{\frac{d_i}{\ln n}}$.
\end{lemma}
\begin{proof}
    Fix such an $i$. Notice that for any $k$ with $i \leq k + i \leq f$ we have $d_{k+i+1} \leq d_{k+i} - 0.9\lambda_{k+i}$ (\cref{lemma:aux_lemma}) and $0.9 \lambda_{k+i} = 0.9 \cdot d_{k+i}^{2/3} \ln^{1/3} n \geq \ceil{0.5 \cdot d_{k+i}^{2/3} \ln^{1/3} n}$ for $n, d_{k+i} \geq 10$. Therefore:
    \begin{equation*}
        d_{k+i+1} \leq d_{k+i} - \ceil{0.5 \cdot d_{k+i}^{2/3} \ln^{1/3} n} \text{  for any $k$ with $i \leq k + i \leq f$.}
    \end{equation*}
    By \Cref{lemma:aux_local_result}, the sequence $(d_{k+i})_{k \geq 0}$ will drop below $1$ after at most $k_{\text{max}} := 6 \sqrt[3]{\frac{d_i}{\ln n}} $ steps. However, $f$ is defined as the highest index $k + i$ for which $d_{k + i} \geq 10^{24} \ln n \geq 1$. This implies that $f \leq i + k_{\text{max}} \leq i + 6 \sqrt[3]{\frac{d_i}{\ln n}}$. In particular: $f - i + 1 \leq 6 \sqrt[3]{\frac{d_i}{\ln n}} + 1 \leq 7 \sqrt[3]{\frac{d_i}{\ln n}}$.
\end{proof}
Finally, we can now upper bound the slacks $a_i$:
\begin{lemma}[Upper bounding slacks $a_i$] \label{lemma:bounding_slacks}
    For $d_0 \geq 1$, define the \emph{slack sequence} $Sl(d_0)$ according to \cref{def:slck_seq}. We have, for any $i \in \{0,\dots,f\}$:
    \begin{equation}
        a_i \leq 10^{24} \ln n + 10^4 \cdot \left( d_i^{3/4} \ln^{1/2} d_i + d_i^{2/3} \ln^{1/3} n \right).
    \end{equation}
\end{lemma}
\begin{proof}
    We begin by recalling \eqref{eq:aux_slack_upper_bound}:
    \begin{equation*}
        a_i \leq 10^{24} \ln n +  (f - i + 1) \cdot \left(2 \lambda_i \cdot \frac{q_i}{d_i} + 2 \lambda_i \cdot \frac{\lambda_i}{d_i} + 16\sqrt{\lambda_i \ln n} \right).
    \end{equation*}
    Replacing $\lambda_i = d_i^{2/3} \ln^{1/3} n$ and applying \cref{lemma:bounding_num_phases}, we get:
    \begin{equation*}
        a_i \leq 10^{24} \ln n + 14 q_i + 14 \cdot d_i^{2/3} \ln^{1/3} n + 112 \cdot d_i^{2/3} \ln^{1/3} n. 
    \end{equation*}
    Replacing $q_i = 10^2 \cdot d_i^{3/4} \ln^{1/2} d_i$ gives:
    \begin{align*}
        a_i &\leq 10^{24} \ln n + 14 \cdot 10^2 \cdot d_i^{3/4} \ln^{1/2} d_i + 126 \cdot d_i^{2/3} \ln^{1/3} n \\
        &\leq 10^{24} \ln n + 10^4 \cdot \left( d_i^{3/4} \ln^{1/2} d_i + d_i^{2/3} \ln^{1/3} n \right).
    \end{align*}
    which is the claimed statement.
\end{proof}

\paragraph{Putting everything together.} By the above discussion, we can, with high probability, color a graph $G$ of known (upper bound on the) maximum degree $d_0 := \Delta$ arriving online edge-by-edge assuming some mild conditions which we now discuss. In the following, we consider the \emph{degree sequence} $D(d_0) := \{d_0,\dots,d_{f+1}\}$ and the \emph{slack sequence} $Sl(d_0) := \{a_0,\dots,a_{f+1}\}$ as defined in \cref{def:deg_seq} and \cref{def:slck_seq}.
\begin{itemize}
    \item By \cref{lemma:deg_seq_property}, for any phase $i$ during the execution of \cref{alg:coloring_generic}, the currently uncolored subgraph $U_i$ has maximum degree at most $d_i$ (see \cref{lemma:deg_seq_property}).

    \item For every arriving edge $e$ of $G$, \cref{alg:coloring_generic} expects to be provided a list $L(e)$ of available colors online. It is required that $|L(e)| \geq 2 \cdot 10^{24} \ln n$ (see \cref{lemma:slck_seq_property}). Furthermore, the algorithm partitions (online) the set of colors $\mathcal{C} := \cup_{e \in G} L(e)$ into $C_0,\dots,C_{f+1}$, and this partitioning is \emph{admissible} as defined in \cref{def:admissible_partitioning}. 
    
    \item For any phase $i \in \{0,\dots,f\}$ in which $e$ is connected to a dense vertex of $U_i$, i.e., $\text{deg}_{U_i}(v) \geq d_i - \lambda_i$, the algorithm uses the sublist $\ell^i(e) := L(e) \cap C_i(e)$, made up of currently unused colors. By the \emph{admissible} choice of the partitioning of colors, it is guaranteed that $\lambda_i \leq |\ell^i(e)| \leq \lambda_i + 10 \sqrt{\lambda_i \ln n}$.
    
    \item If, during some phase $i$ of \cref{alg:coloring_generic}, the edge $e$ is connected to a dense vertex $v$ in $U_i$, i.e.\ $\text{deg}_{U_i}(v) \geq d_i - \lambda_i$, the list $L^i(e)$ must have size at least $|L^i(e)| \geq d_i + a_i$. By \cref{lemma:slck_seq_property}, it suffices that $|L^i(e)| \geq d_i + a_i$ holds for the first phase $i$ in which $e$ is connected to a dense vertex.
\end{itemize}

We sum up the discussion above in a concise theorem:
\begin{theorem}[Main Coloring Theorem] \label{thm:main_coloring}
    Assume the input for \Cref{alg:coloring_generic} is such that every arriving edge $e$ presents a list $L(e)$ with $|L(e)| \geq 2 \cdot 10^{24} \ln n$ colors. Furthermore, the input contains a partitioning of the set of colors $\mathcal{C} := \cup_{e \in G} L(e)$ into $C_0,\dots,C_{f+1}$, satisfying:
    \begin{enumerate}
        \item \label{cond_1_of_main_coloring} 
        The partitioning $C_0,\dots,C_{f+1}$ of the colors is \emph{admissible} as defined in \cref{def:admissible_partitioning}.

        \item \label{cond_2_of_main_coloring}
        For any edge $e = (u,v)$, we have:
        \begin{equation} \label{eq:main_cond_in_main_coloring}
            |L(e)| \geq d_{i(e)} + 10^{24} \ln n + 10^4 \cdot \left( d_{i(e)}^{3/4} \ln^{1/2} d_{i(e)} + d_{i(e)}^{2/3} \ln^{1/3} n \right),
        \end{equation}
        where:
        \begin{equation}
            i(e) := \max \{ i \in \{0,\dots,f+1\} : d_i \in D(d_0), \text{$d_i \geq \deg_G(u)$ \emph{and} $d_i \geq \deg_G(v)$}\}.
        \end{equation}
    \end{enumerate}
    Under the guarantee of such an input, \cref{alg:coloring_generic} colors all edges $e$ of $G$ online with high probability (in $n$), such that every edge $e$ is assigned a color $c \in L(e)$.
\end{theorem}

\begin{proof}
    By the discussion preceding the theorem statement, it suffices to prove that the lists $L^i(e)$ are large enough across all phases $i$, such that the statement follows by \cref{lemma:deg_seq_property} and \cref{lemma:slck_seq_property}. More concretely, as discussed previously and implied by \cref{lemma:slck_seq_property}, it suffices to prove for any edge $e$ that $|L^i(e)| \geq d_i + a_i$ holds for the first phase $i$ in which $e$ is connected to a dense vertex. We claim that condition 2 from the theorem statement implies this fact. Indeed, first let:
    \begin{equation}
        d_{\text{max}}(e) := \max\{ \text{deg}_G (u) : e \in u \}
    \end{equation}
    be the degree of the largest-degree vertex connected to $e$ in the original graph $G$. Furthermore, consider $i(e)$ as defined in the theorem statement be the largest index $\leq f + 1$ in the \emph{degree sequence} of $d_0$ for which $d_{i(e)} \geq d_{\text{max}}(e)$. In particular, the index of the first phase $i$ in which $e$ is connected to a dense vertex in $U_i$ is at least $i(e)$. By \cref{lemma:slck_seq_property}, if the initial list $L(e)$ provided for $e$ has size $|L(e)| \geq d_{i(e)} + a_{i(e)}$, then for all phases $j \geq i(e)$ in \cref{alg:coloring_generic} we will have $|L^j(e)| \geq d_j + a_j$. To end the proof, notice that by \cref{lemma:bounding_slacks}, the imposed condition \eqref{eq:main_cond_in_main_coloring} implies $|L(e)| \geq d_{i(e)} + a_{i(e)}$.
\end{proof}

%% file: Sections/list-coloring.tex
\subsection{List Edge Coloring} \label{sec:list-coloring}

In this section we consider online \emph{list} edge coloring, where we recall that each edge $e=(u,v)$ arrives with a list $L(e)\subseteq \mathbb{N}$, and can only be colored using a color from $L(e)$ (while again guaranteeing no vertex has more than one edge of any color).
Our main result for this problem is the following.

\ListColoringTheorem*

\begin{proof}
We prove \cref{thm:list_edge_coloring} by running \cref{alg:coloring_generic} and applying \cref{thm:main_coloring}. For this purpose, we design an online algorithm $\mathrm{Input}$ such that conditions \ref{cond_1_of_main_coloring} and \ref{cond_2_of_main_coloring} are fulfilled. First, by assumption, all provided lists $L(e)$ have size $|L(e)| \geq \Delta + q$, which already guarantees condition \ref{cond_2_of_main_coloring} of \cref{thm:main_coloring}. Hence, it remains to design an online algorithm which partitions $\mathcal{C} := \cup_{e \in G} L(e)$ into $C_0,\dots,C_{f+1}$ such that condition \ref{cond_1_of_main_coloring} is fulfilled (i.e.\ \emph{admissibility} of the partitioning as defined in \cref{def:admissible_partitioning}).

As an additional point, we run \cref{alg:coloring_generic} with a minor change, which can be made without loss of generality. Concretely, before the start of any phase $i \in \{0,\dots,f\}$, if any list $L^i(e)$ of remaining colors for some edge $e$ has size $|L^i(e)| > d_i + a_i$, we prune some of the extra colors (chosen arbitrarily) in the list such that $|L^i(e)| = d_i + a_i$ holds. Effectively these extra colors are simply ignored.

Now we describe the \emph{admissible} partitioning of colors. We define $C_0,\dots,C_f$ as follows: For any $i \in \{0,\dots,f\}$ let $\mathcal{C}^i := \mathcal{C} \setminus \cup_{0 \leq j < i} \ C_j$ be the set of not (yet) used colors before phase $i$. Construct $C_i$ by sampling (online) any color $c \in \mathcal{C}^i$ with probability $p_i := (\lambda_i + 5\sqrt{\lambda_i \ln n})/(d_i + a_i)$. We argue that $p_i$ is a valid probability, i.e., $p_i \leq 1$. Indeed, this inequality is implied by $d_i \geq \lambda_i + 5\sqrt{\lambda_i \ln n}$. To show this, notice that by \cref{lemma:aux_lemma}, we have $5 \sqrt{\lambda_i \ln n} \leq 0.1 \lambda_i$, such that the previous inequality is implied by $d_i \geq 1.1 \lambda_i$, or $d_i \geq (1.1)^2 \ln n$ which is obviously true.

It is clear that this sampling algorithm can be implemented online. Whenever a color $c \in L(e)$ is seen for the first time, in the list $L(e)$ of some arriving edge $e$, we iterate through $i \in \{0,\dots,f\}$ and assign $c$ to $C_i$ with corresponding probability $p_i$ (and stop the algorithm as soon as $c$ is assigned to some such $C_i$). Finally, let $C_{f+1}$ contain all those colors that have been assigned to none of $C_0,\dots,C_f$.

We claim that the above partitioning algorithm provides an \emph{admissible} partitioning as defined in \cref{def:admissible_partitioning} and as required by condition \ref{cond_1_of_main_coloring} of \cref{alg:coloring_generic}. In fact, we will prove something slightly stronger by arguing the admissibility condition holds for \emph{all} edges $e$, and not only those which are connected to a dense vertex. In order to do this, fix a phase $0 \leq i \leq f$ and assume that condition \ref{cond_1_of_main_coloring} was fulfilled for all previous phases and edges in the past. It remains to prove that the sublists induced by the algorithm, $\ell^i(e) := L(e) \cap C^i(e)$, have the required size for all edges $e$:
\begin{lemma} \label{lemma:helper_lemma_list_edge_coloring}
    The above partitioning algorithm induces sublists $\ell^i(e) := L^i(e) \cap C_i(e)$, which satisfy with probability at least $1 - \frac{1}{n^5}$:
    \begin{equation*}
        \lambda_i \leq |\ell^i(e)| \leq \lambda_i + 10 \sqrt{\lambda_i \ln n} \text{ for all $e$.}
    \end{equation*}
\end{lemma}
\begin{proof}[Proof of \cref{lemma:helper_lemma_list_edge_coloring}]
    Fix an edge $e$ and its list of available colors $L^i(e)$. Since condition \ref{cond_1_of_main_coloring} of \cref{thm:main_coloring} is satisfied for all edges $e$ until the current point in the execution of \cref{alg:coloring_generic}, we have $|L^i(e)| \geq d_i + a_i$. By our slight modification of \cref{alg:coloring_generic}, we may assume equality, $|L^i(e)| = d_i + a_i$.
    
    Hence, every color $c \in L^i(e)$ is picked by the sampling algorithm with probability $p := (\lambda_i + 5\sqrt{\lambda_i \ln n})/|L^i(e)|$, such that the expected number $X$ of chosen colors is distributed $X \sim \text{Bin}(|L^i(e)|, p)$ with $\mu := \E[X] = \lambda_i + 5 \sqrt{\lambda_i \ln n}$. By a standard Chernoff bound, we have for any $\varepsilon \in [0,1]$:
    \begin{equation*}
        \Pr[ |X - \mu| \geq \varepsilon \mu ] \leq 2 \exp \left( -\frac{\varepsilon^2 \cdot \lambda_i}{3} \right).
    \end{equation*}
    Fix $\varepsilon := \sqrt{18 \cdot \frac{\ln n}{\mu}}$. As $d_i \geq 10^{24} \ln n$, we have that $\mu \geq \lambda_i = d_i^{2/3} \ln^{1/3} n \geq 18 \ln n$, and so $\varepsilon \in [0,1]$. Furthermore , we have $\varepsilon \mu \leq 5 \sqrt{\lambda_i \ln n}$. Indeed, this is equivalent (by squaring) to:
    \begin{align*}
        18 \ln n \cdot (\lambda_i &+ 5 \sqrt{\lambda_i \ln n}) \leq 25 \lambda_i  \ln n \\
        90\sqrt{\lambda_i \ln n} &\leq 7 \lambda_i.
    \end{align*}
    This last expression is equivalent to $d_i \geq (90/7)^4 \ln n$ which is obvious. Hence, the Chernoff inequality gives:
    \begin{equation*}
        \Pr[ |X - \mu| \geq 5 \sqrt{\lambda_i \ln n}] \leq 2 \cdot \frac{1}{n^{6}} \leq \frac{1}{n^5},
    \end{equation*}
    and so $X = |\ell^i(e)| \in [\lambda_i, \lambda_i + 10 \sqrt{\lambda_i \ln n}]$ with probability at least $1 - \frac{1}{n^5}$ as claimed.
\end{proof}

Now that the lemma is proven, it follows easily by union bound that, with high probability in $n$, all induced sublists $\ell^i(e)$ in all phases $i \in \{0,\dots,f\}$ have the required size, and so condition \ref{cond_1_of_main_coloring} of \cref{thm:main_coloring} is fulfilled. This finishes the proof of \cref{thm:list_edge_coloring}.
\end{proof}

\paragraph{An improvement of \cref{lemma:reduction2}.} As a final note, we observe that the slack $q$ in \cref{thm:list_edge_coloring} can be naturally decomposed into two parts: the first part is the $O(\Delta^{3/4}\log^{1/2}\Delta)$-term, which comes directly from the application of \cref{thm:reducing_max_degree_generic}, which in turn relies on the fact that---by \cref{thm:online_matching}---we have access to an online matching algorithm that matches any edge $e$ with probability $1 / (\Delta + \Theta(\Delta^{3/4}\log^{1/2}\Delta))$. The second part of the slack is the $O(\Delta^{2/3}\log^{1/3}n)$-term, which comes from our choice of $\lambda$ in \cref{alg:coloring_generic}.

However, it is not hard to see that these two parts of the final slack $q$ in \cref{thm:list_edge_coloring} arise independently of each other and, more generally, if one had access to an online matching algorithm with a different guarantee than $1 / (\Delta + \Theta(\Delta^{3/4}\log^{1/2}\Delta))$ matching probability per edge, this would directly translate to a change in the corresponding first term of the final slack $q$. More concretely, for the classical online edge coloring problem, we can generalize \cref{lemma:reduction2} to obtain the following reduction from online edge coloring to online matching:

\begin{lemma}[Improved Reduction]\label{lemma:reduction_improved_appendix}
    Let $\calA$ be an online matching algorithm that, on any graph of maximum degree $\Delta = \omega(\log n)$, matches each edge with probability at least $1 / (\alpha \cdot \Delta)$, for $\alpha \geq 1$. 
    Then, there exists an online edge coloring algorithm $\calA'$ that on any graph with maximum degree $\Delta = \omega(\log n)$ outputs an edge coloring with $(\alpha + O((\log n / \Delta)^{1/3})) \cdot \Delta$ colors with high probability in $n$.
\end{lemma}

%% file: Sections/local-result.tex
\subsection{Local Edge Coloring} \label{sec:local-coloring}

In this section we consider online \emph{local} edge coloring, where we recall that we wish to color each edge $e=(u,v)$ with a color not much higher than $d_{\text{max}}(e) := \max\{\deg(u),\deg(v)\}$. Our main result for this problem is the following.

\LocalColoringTheorem*

\begin{proof}
The statement of \cref{thm:local_edge_coloring_strenghtened} is almost an immediate consequence of our more general \cref{thm:main_coloring}. For $i \in \{0,\dots,f+1\}$, let $d_i$ and $a_i$ denote the entries from the \emph{degree sequence} and \emph{slack sequence} of $d_0 := \Delta(G)$ as defined in \cref{def:deg_seq} and \cref{def:slck_seq} (also see \cref{thm:main_coloring}). We define the set of colors $\mathcal{C}$ to be $\mathcal{C} := \{1,\dots,d_0 + a_0\}$ and propose the following partitioning:
\begin{align*}
    C^i &:= \{d_i + a_i - (\lambda_i - 1), \dots, d_i + a_i\} \text{ for $i \in \{0,\dots,f\}$} \\
    C^{f+1} &:= \{1,\dots,2 \cdot d_f \}.
\end{align*}
The fact that this is a valid partitioning follows from \cref{lemma:slck_seq_property}.

Now fix an edge $e =(u,v)$ and let $d_{\text{max}}(e) := \max\{\deg(u), \deg(v)\}$. If $d_{\text{max}}(e) \leq d_{f+1}$, consider the following input list $L(e)$ of available colors:
\begin{equation}
    L(e) := \{1, \dots, 2 \cdot d_{f+1}\}.
\end{equation}
As $e$ is never connected to a dense vertex during phases $i \in \{0,\dots,f\}$, condition \ref{cond_1_of_main_coloring} is vacuously true. Furthermore, $e$ can be assigned in phase $f+1$ a color $c(e) \in \ell^{f+1}(e) = L(e) \cap C^{f+1} = L(e)$, such that $c(e) \leq 2 \cdot d_{f+1} \leq 2 \cdot 10^{24} \ln n$, which implies the statement from \cref{thm:main_coloring} for edge $e$.

In the following, assume $d_{\text{max}}(e) > d_{f+1}$. Define $i(e)$ as in \cref{thm:main_coloring} and notice that $i(e) \leq f$ because $d_{\text{max}}(e) > d_{f+1}$. Consider the following input list $L(e)$:
\begin{equation} \label{eq:construction_of_Le_local_edge_coloring}
    L(e) := \left\{ 1, \dots, d_{i(e)} + 10^{24} \ln n + 10^4 \cdot \left( d_{i(e)}^{3/4} \ln^{1/2} d_{i(e)} + d_{i(e)}^{2/3} \ln^{1/3} n \right) \right\}.
\end{equation}
Condition \ref{cond_2_of_main_coloring} of \cref{thm:main_coloring} is trivially fulfilled. Consider the induced partitioning of $L(e)$ into sublists $\ell^0(e),\dots,\ell^{f+1}(e)$, where $\ell^{i}(e) := L(e) \cap C^i$. It holds that $|\ell^i(e)| = \lambda_i$ for any phase $i \leq f$ in which edge $e$ is connected to a dense vertex of $U_i$. This is because, for all such $i$, we have $C^i \subseteq L(e)$ by \cref{lemma:bounding_slacks}, and $|C^i| = \lambda_i$. This property guarantees condition \ref{cond_1_of_main_coloring} of \cref{thm:main_coloring}.

By \cref{thm:main_coloring}, we obtain that, with high probability, \cref{alg:coloring_generic} assigns each edge $e$ a color $c(e)$ satisfying $c(e) \leq |L(e)|$.
Recall $d_{\text{max}}(e) := \max\{\deg(u), \deg(v)\} > d_{f+1}$, which implies $i(e) \leq f$. First, we will make use the following inequality, whose proof is deferred:
\begin{lemma} \label{lemma:helper_lemma_local_edge_coloring_1}
    With the above notations, and assuming $i(e) \leq f$, it holds that:
    \begin{equation}
        d_{i(e)} \leq d(e) := d_{\text{max}}(e) + 2 \cdot d_{\text{max}}^{2/3}(e) \ln^{1/3} n.
    \end{equation}
\end{lemma}
Combining with the definition of $L(e)$ in \eqref{eq:construction_of_Le_local_edge_coloring} and with the fact that $c(e) \leq |L(e)|$, we obtain by the above lemma that:
\begin{equation} \label{eq:upper_bound_c(e)_1}
    c(e) \leq |L(e)| \leq g(e) := d(e) + 10^{24}\ln n + 10^4 \cdot \left( d(e)^{3/4} \ln^{1/2} d(e) + d(e)^{2/3} \ln^{1/3} n \right).
\end{equation}
To finish the proof of the theorem, it suffices to show that:
\begin{lemma} \label{lemma:helper_lemma_local_edge_coloring_2}
    With $g(e)$ as defined above, we have:
    \begin{equation} \label{eq:upper_bound_c(e)_2}
        g(e) \leq d_\text{max}(e) + 2 \cdot 10^{24} \ln n + 10^5 \cdot \left( d_{\text{max}}^{3/4}(e) \ln^{1/2} d_{\text{max}}(e) + d_{\text{max}}^{2/3}(e) \ln^{1/3} n \right).
    \end{equation}
\end{lemma}
As in the case of \cref{lemma:helper_lemma_local_edge_coloring_1}, we defer the proof (see below). By combining \eqref{eq:upper_bound_c(e)_1} and \eqref{eq:upper_bound_c(e)_2}, we get the desired upper bound for $c(e)$, and the statement of \cref{thm:local_edge_coloring_strenghtened} follows.
\end{proof}

We now present the proofs of lemmas \cref{lemma:helper_lemma_local_edge_coloring_1} and \cref{lemma:helper_lemma_local_edge_coloring_2}, which were deferred previously:
\begin{proof}[Proof of \cref{lemma:helper_lemma_local_edge_coloring_1}]
    By definition, $i(e)$ is the largest index $i \in \{0,\dots,f+1\}$, such that $d_{i(e)} \geq d_{\text{max}}(e)$. As $i(e) \neq f+1$, it follows that $d_{i(e) + 1} \leq d_{\text{max}}(e)$, and from \cref{def:deg_seq} we have $d_{i(e) + 1} \geq d_{i(e)} - \lambda_{i(e)}$. Recalling that $\lambda_{i(e)} = d^{2/3}_{i(e)} \ln^{1/3} n$, we conclude:
    \begin{equation} \label{eq:aux_local_application_2}
        d_{\text{max}}(e) \geq d_{i(e)} - \lambda_{i(e)} = 
        d_{i(e)} \cdot \left( 1 - \sqrt[3]{\frac{\ln n}{d_{i(e)}}} \right) \geq d_{i(e)} \cdot \left( 1 - \sqrt[3]{\frac{\ln n}{d_{\text{max}}(e)}} \right).
    \end{equation}
    Notice that, as $i(e) \leq f$, we have $d_{i(e)} \geq 10^{24} \ln n$, which implies $d_{\text{max}}(e) \geq d_{i(e)} - \lambda_{i(e)} > 4 \ln n$, such that $\sqrt[3]{\frac{\ln n}{d_{\text{max}}(e)}} < 1 / 2$. Therefore, inequality \eqref{eq:aux_local_application_2} gives:
    \begin{align*}
        d_{i(e)} & \leq d_{\text{max}}(e) \cdot \frac{1}{1 - \sqrt[3]{\frac{\ln n}{d_{\text{max}}(e)}}} \leq d_{\text{max}}(e) \cdot \left( 1 + 2 \cdot \sqrt[3]{\frac{\ln n}{d_{\text{max}}(e)}} \right) = d_{\text{max}(e)} + 2 \cdot d_{\text{max}}^{2/3}(e) \ln^{1/3} n. \qedhere
    \end{align*}
\end{proof}

\begin{proof}[Proof of \cref{lemma:helper_lemma_local_edge_coloring_2}]
    Recall that $g(e)$ was defined as:
    \begin{equation} \label{eq:helper_bound_g(e)}
        g(e) := d(e) + 10^{24}\ln n + 10^4 \cdot \left( d(e)^{3/4} \ln^{1/2} d(e) + d(e)^{2/3} \ln^{1/3} n \right),
    \end{equation}
    where $d(e) = d_{\text{max}}(e) + 2 \cdot d_{\text{max}}^{2/3}(e) \ln^{1/3} n$. Now, if $d_{\text{max}}(e) \leq \ln n$, then  $d(e) \leq 3 \ln n$. But then \eqref{eq:helper_bound_g(e)} immediately gives that $g(e) \leq 2 \cdot 10^{24} \ln n$, which implies the lemma.
    Conversely, if $d_{\text{max}}(e) \geq \ln n$, then $d(e) \leq 3 \cdot d_{\text{max}}(e)$. Therefore, we have:
        \begin{align*}
            d(e)^{3/4} \ln^{1/2} d(e) &\leq 10 \cdot d_{\text{max}}^{3/4}(e) \ln^{1/2} d_{\text{max}}(e) \\
            d(e)^{2/3} \ln^{1/3} n &\leq 3 \cdot d_{\text{max}}^{2/3}(e) \ln^{1/3} n.
        \end{align*}
        But then, by \eqref{eq:helper_bound_g(e)}, we obtain the desired inequality:
        \begin{align*}
            g(e) &\leq d_{\text{max}}(e) + 2 \cdot d_{\text{max}}^{2/3}(e) \ln^{1/3} n + 10^{24}\ln n \ + 10^4 \cdot \left( 10 \cdot d_{\text{max}}^{3/4}(e) \ln^{1/2} d_{\text{max}}(e) + 3 \cdot d_{\text{max}}^{2/3}(e) \ln^{1/3} n \right) \\
            &\leq d_{\text{max}}(e) +  10^{24}\ln n + 10^5 \cdot \left( d_{\text{max}}^{3/4}(e) \ln^{1/2} d_{\text{max}}(e) + d_{\text{max}}^{2/3}(e) \ln^{1/3} n \right). \qedhere
        \end{align*}
\end{proof}

%% file: Sections/extension.tex
\section{Extension: Online Rounding of Fractional Matchings} \label{sec:extension}

In this section we generalize the result in \cref{sec:analysis} to a rounding algorithm for fractional matchings. By \emph{fractional matching}, we mean an assignment $x\in [0,1]^{E}$ to the edges such that for each vertex $v$, we have $\sum_{e\in \delta(v))} x_e \le 1$. In the online model, the value of $x_e$ is revealed alongside the edge $e$. An online rounding algorithm must decide immediately and irrevocably whether to match the newest arriving edge. The objective is to match each edge $e$ with probability close to $x_e$.

\Cref{alg:edge-arrival} can be viewed as an online algorithm which rounds the uniform fractional matching $\{x_e = \frac{1}{\Delta}\}_{e\in E}$ to an integral one, while nearly preserving marginals, $\Pr[e\text{ matched}] = \frac{1}{\Delta}(1-o(1))$. We show that one can---in a straightforward manner---generalize our \cref{alg:edge-arrival} to round arbitrary fractional matchings (also in non-bipartite graphs, see \cref{rem:non-bipartite}), as long as they are sufficiently ``spread out'', resulting in our \cref{alg:edge-arrival_extension,thm:ronding_matchings_theorem}. In particular, as long as all $x_e\le \eps$, we will be able to round $\vec{x}$ while only incurring some a small loss which goes to zero when $\eps\to 0$. To our knowledge, this is the first online rounding algorithm with such guarantees for general (non-bipartite) graphs, and also for adversarial edge arrivals in any (bipartite) graph.

\begin{figure}[h!]
\begin{center}
\begin{minipage}{0.95\textwidth}
\begin{mdframed}[hidealllines=true, backgroundcolor=gray!15]

\begin{algorithm}[\edgearrivalalgext] \ \\[0.2cm]
\emph{At the arrival of edge $e_t = (u,v)$ at time $t$:}
\item  Sample $X_t \sim [0,1]$ uniformly at random.
    \item Define 
    \begin{align*}
        P(e_t) & = \begin{cases}
         x_e \cdot (1 - s(\varepsilon)) \cdot 
         \frac{1}{F_t(u) \cdot F_t(v)} 
         & \mbox{if $u$ and $v$ are unmatched in $M_t$,} \\ 
         0 & \mbox{otherwise.}
         \end{cases}
         \intertext{and}
        \hat P(e_t) & = \begin{cases}
        P(e_t) & \mbox{if $\min\{F_t(u), F_t(v)\}\cdot (1-P(e_t))\geq \frac{s(\varepsilon)}{4}$} \\ 
         0 & \mbox{otherwise.}
         \end{cases} 
    \end{align*}
    
    \item Set 
    \begin{itemize}
        \item $F_{t+1}(u) \leftarrow F_t(u) \cdot (1- \hat P(e_t))$;
        \item $F_{t+1}(v)\leftarrow F_t(v) \cdot (1-\hat P(e_t))$; 
        \item $M_{t+1} = \begin{cases} M_t \cup \{e_t\} & \mbox{ if $X_t < \hat P(e_t)$,} \\
        M_t & \mbox{otherwise.}
        \end{cases}$
    \end{itemize}
\label{alg:edge-arrival_extension}
\end{algorithm}
\end{mdframed}
\end{minipage}
\end{center}
\end{figure}

\begin{theorem} \label{thm:ronding_matchings_theorem}
    Let $x\in \mathbb{R}^{E}$ be a fractional matching of some graph $G$, which is revealed online, and satisfies $x_e \leq \varepsilon$ for all edges $e$, for some known $\varepsilon \leq 0.99$. Then there exists a randomized online matching algorithm whose output matching $\calM$ satisfies for any edge $e$:
    \begin{equation*}
        x_e \geq \Pr[\textup{$e$ matched by } \calM] \ge x_e \cdot (1 - s(\varepsilon)),
    \end{equation*}
    where $s(\varepsilon) := C \cdot \sqrt[4]{\varepsilon} \cdot \sqrt[2]{\left(\ln \frac{1}{\varepsilon}\right)}$ (for some large enough constant $C>0$) converges to $0$ as $\varepsilon \rightarrow 0$.
\end{theorem}

\begin{remark}\label{rem:non-bipartite}
Note that in general graphs, fractional matchings can be $3/2$ times larger than their largest integral counterparts, as exemplified by a triangle with values $x_e=1/2$ for all its edges. That is, the integrality gap of this relaxation of matchings is $3/2$.
Nonetheless, our algorithm works for non-bipartite graphs, and for sufficiently spread out fractional matchings we almost losslessly round them to integral matchings. 
This does not contradict the integrality gap of this relaxation in general graphs, as all ``odd set'' constraints in the integral matching polytope \cite{edmonds1965maximum} are approximately satisfied for spread out fractional matching $x_e\le \eps$. 
On the other hand, to round fractional matchings in non-bipartite graphs, it is necessary to incur some loss (with respect to $\eps$) when rounding, even in offline settings.
\end{remark}

The new algorithm, a straightforward adaptation of \cref{alg:edge-arrival}, is given by \cref{alg:edge-arrival_extension}. The following results are proven analogously to their counterparts from \cref{sec:analysis}, and are therefore omitted here:

\begin{observation}[Corresponds to \Cref{obs:Ftv_and_Pet_bounds}]
\label{obs:analogue_Ftv_and_Pet_bounds}
    $F_t(v) \geq \frac{s(\varepsilon)}{4}$ and $\hat P(e_t) \leq P(e_t) \leq \frac{\varepsilon}{s^2(\varepsilon)}$ for every vertex $v\in V$ and time $t$.
\end{observation}

\begin{observation}[Corresponds to \Cref{obs:fpp-det-by-M}]
\label{obs:analogue_fpp-det-by-M}
    For any $t$, the random variables $F_{t}(v), P(e_t), \hat P(e_t)$ are determined by the current partial input $e_1,\dots,e_t$ and the current matching $M_{t-1}$.
\end{observation}

\begin{lemma}[Corresponds to \Cref{lemma:sufficient_cond_for_right_marginal}]
\label{lemma:analogue_sufficient_cond_for_right_marginal}
    For any edge $e_t$ it holds that $\Pr[X_t < P(e_t)] = x_e \cdot (1-s(\varepsilon))$.
\end{lemma}

\begin{lemma}[Corresponds to \Cref{lemma:sufficient_cond_for_Phat_equal_P}]
\label{lemma:analogue_sufficient_cond_for_Phat_equal_P}
    If $e_t =(u,v)$ and $\min\{F_{t}(u),\ F_{t}(v)\} \geq s(\varepsilon)/3$, then $\hat P(e_t) = P(e_t)$.
\end{lemma}
The proof of \Cref{lemma:analogue_sufficient_cond_for_Phat_equal_P} (just as its counterpart \Cref{lemma:sufficient_cond_for_Phat_equal_P}) requires $P(e_t) \leq \frac{\varepsilon}{s^2(\varepsilon)} \leq 1/4$. This can be achieved by imposing:
\begin{equation} \label{eq:setting_C}
    C \geq \max_\varepsilon \left( \frac{4 \varepsilon}{\sqrt{\varepsilon} \cdot \ln \frac{1}{\varepsilon}} \right).
\end{equation}
However, as $C$ is supposed to be a constant, we still need to check that the right hand side of the above inequality is also a constant. By the statement of \Cref{thm:ronding_matchings_theorem}, we have that $\varepsilon \leq 0.99$, and it is easy to check that the function $\varepsilon \mapsto \left( \frac{\varepsilon}{\sqrt{\varepsilon} \cdot \ln \frac{1}{\varepsilon}} \right)$ is bounded in the interval $(0, 0.99)$. Hence, the right hand side of \eqref{eq:setting_C} is a constant.

The analogue of \Cref{lemma:main_ineq} is:
\begin{lemma}[Corresponds to \Cref{lemma:main_ineq}]
\label{lemma:analogue_main_ineq}
    Let $e_{t_1} = (u_1,v),\ldots, e_{t_\ell}  =(u_{\ell},v)$ be the edges incident to $v$ with $t_1< \cdots < t_\ell$. Further, let $S := \{u_i \mid u_i \not \in M_{t_i}\}$ be those neighbors $u_i$ that are unmatched by time $t_i$ when edge $e_{t_i} =(u_i,v)$ arrives. If:
    \begin{align}
        \sum_{u_i \in S} \frac{x_{e_{t_i}} \cdot (1-s(\varepsilon))}{F_{t_i}(u_i)} \leq 1 - \frac{s(\varepsilon)}{3},
    \end{align}
    then $F(v) \geq s(\varepsilon)/3$. 
\end{lemma}

\noindent For ease of notation let $x_i := x_{e_{t_i}}$ be the fractional input of the edge $e_{t_i} =(u_i,v)$, which connects $v$ to its neighbor $u_i$. We will derive a martingale from the following quantities:
\[
    S_t := \{ u_i \in N(v) \mid u_i \not \in M_{\min\{t, t_i\}}\} \quad \mbox{and} \quad Y_{t-1} := \sum_{u_i \in S_t} \frac{x_i \cdot (1 - s(\varepsilon))}{F_{\min\{t,t_i\}}(u_i)}.
\]
\begin{claim}[Corresponds to \Cref{lemma:proof_y_is_martingale}] \label{lemma:analogue_proof_y_is_martingale}
    $Y_0, \dots, Y_m$ form a martingale w.r.t.\ the random variables $X_1,\dots,X_m$. Furthermore, the difference $Y_t - Y_{t-1}$ is given by the following two cases:
    \begin{itemize}
        \item If $e_t$ is added to $M_{t+1}$, which happens with probability $\hat P(e_t)$, then:
        \begin{equation}
            Y_{t} - Y_{t-1}  = - \sum_{u_i \in S_t \cap e_t} \frac{x_i \cdot (1 - s(\varepsilon))}{F_{t}(u_i)}.
        \end{equation}
        \item If instead $e_t$ is not added to $M_{t+1}$, which happens with probability $1 - \hat P(e_t)$, then:
        \begin{equation}
            Y_{t} - Y_{t-1}  = \frac{\hat P(e_t)}{1 - \hat P(e_t)} \cdot \sum_{u_i \in S_t \cap e_t} \frac{x_i \cdot (1 - s(\varepsilon))}{F_{t}(u_i)}.
        \end{equation}
    \end{itemize}
\end{claim}

To apply Freedman's inequality, as in Section \ref{sec:analysis}, bounds on the step size and on the variance of the martingale are required. They are given by the following two lemmas:
\begin{lemma}[Corresponds to \Cref{lemma:bounding_step_size}] 
    For all times $t$ and realization of the randomness, $|Y_{t} - Y_{t-1}| \leq A$, where $A:=\frac{8 \varepsilon}{s(\varepsilon)}$.
\end{lemma}
\begin{proof}
    By using the expressions for the difference $Y_t - Y_{t-1}$ from \cref{lemma:analogue_main_ineq}, we obtain:
    \begin{align*}
        &|Y_{t} - Y_{t-1}| \leq \max \left \{ \frac{\hat P(e_t)}{1 - \hat P(e_t)}, 1 \right \} \cdot \sum_{u_i \in S_t \cap e_t} \frac{x_i \cdot (1 - s(\varepsilon))}{F_{t}(u_i)} \leq  \sum_{u_i \in S_t \cap e_t} \frac{\varepsilon}{s(\varepsilon) / 4} \leq \frac{8\varepsilon}{s(\varepsilon)}.
    \end{align*}
    For the second inequality, first notice that we already guarantee $\hat P(e_t) \leq P(e_t) \leq 1/4$ for the proof of \Cref{lemma:analogue_sufficient_cond_for_Phat_equal_P}, so in particular $\frac{\hat P(e_t)}{1 - \hat P(e_t)} \leq 1$. Also, we have $F_t(u_i) \geq s(\varepsilon)/4$ (by \Cref{obs:analogue_fpp-det-by-M}) at any point of time in the algorithm. For the third inequality we used the trivial fact that $|S_t \cap e_t| \leq 2$. \qedhere
\end{proof}

\begin{lemma}[Corresponds to \Cref{lemma:ub_variance_martingale}] 
    Consider the martingale described above. We have:
    \begin{align}
    \sum_{t=1}^m \E[ (Y_t - Y_{t-1})^2 &\mid X_{t-1},\ldots X_1] \leq 128 \ln \left( \frac{4}{s(\varepsilon)} \right) \cdot \frac{\varepsilon}{(s(\varepsilon))^2}.
    \end{align}
\end{lemma}
\begin{proof}
    Mimicking the proof of \cref{lemma:ub_variance_martingale}: Assuming edge $e_t$ arrives at time $t$, we have:
    \begin{align*}
        & \E[ (Y_t - Y_{t-1})^2 \mid X_{t-1},\ldots X_1] \\
        \leq & 2 \hat P(e_t) \cdot \sum_{u_i \in S_t \cap e_t} \left( \frac{x_i \cdot (1 - s(\varepsilon))}{F_t(u_i)} \right)^2 + 2 (1 - \hat P(e_t)) \sum_{u_i \in S_t \cap e_t} \cdot \left( \frac{\hat P(e_t) \cdot x_i \cdot (1 - s(\varepsilon))}{F_t(u_i) (1- \hat P(e_t))} \right)^2 \\
        = & \sum_{u_i \in S_t \cap e_t} \frac{2 \hat P(e_t) \cdot x_i^2 \cdot (1 - s(\varepsilon))^2}{(F_t(u_i))^2} \left( 1 + \frac{\hat P(e_t)}{1- \hat P(e_t)} \right) \\
        \leq & 128 \frac{\hat P(e_t)}{(s(\varepsilon))^2} \cdot x_i^2,
    \end{align*}
    where we used $F_t(u_i) \geq s(\varepsilon)/4$, $1 + \frac{\hat P(e_t)}{1- \hat P(e_t)} \leq 2$ and $|S_t \cap e_t| \leq 2$.

    We sum the above inequality over all $t$. An edge $e_t$ is thereby summed over on the right hand side only if $e_t$ is incident to some vertex $u_i \in S_t$. As $S_t \subseteq S_0 = N(v)$, we have the following upper bound:
    \begin{equation}
        \sum_{t=1}^m \E[ (Y_t - Y_{t-1})^2 \mid X_{t-1},\ldots X_1] \leq \sum_{i = 1}^{\ell} x^2_i \cdot \sum_{e_t \in \delta(u_i)} 128 \frac{\hat P(e_t)}{(s(\varepsilon))^2}.
    \end{equation}
    For any neighbor $u_i$ of $v$ we have:
    \begin{equation*}
        s(\varepsilon)/4 \leq F_m(u_i) = \prod_{e_t \in \delta(u_i)} (1-\hat P(e_t)) \leq \exp \left( -\sum_{e_t \in \delta(u_i)} \hat P(e_t) \right),
    \end{equation*}
    which implies that:
    \begin{equation*}
        \sum_{e_t \in \delta(u_i)} \hat P(e_t) \leq \ln \left( \frac{4}{s(\varepsilon)} \right).
    \end{equation*}
    Hence:
    \begin{equation*}
        \sum_{i = 1}^\ell x^2_i \cdot \sum_{e_t \in \delta(u_i)} 128 \frac{\hat P(e_t)}{(s(\varepsilon))^2} \leq 
        \sum_{i = 1}^\ell x^2_i \cdot 128 \ln \left( \frac{4}{s(\varepsilon)} \right) \cdot \frac{1}{(s(\varepsilon))^2} \leq
        128 \ln \left( \frac{4}{s(\varepsilon)} \right) \cdot \frac{\varepsilon}{(s(\varepsilon))^2},
    \end{equation*}
    because $\sum_{i = 1}^\ell x_i^2 = \sum_{e := \{v,u_i\}} x^2_e \leq \varepsilon \cdot \sum_{e := \{v, u_i\}} x_e \leq \varepsilon$.
\end{proof}

We are finally ready to analyze our online rounding algorithm.
\begin{proof}[Proof of \cref{thm:ronding_matchings_theorem}]
\noindent First, we remark that the proof of the upper bound is straightforward, as we have:
\begin{equation*}
    \Pr[\text{$e_t$ matched}] = \Pr[X_t < \hat P(e_t)] \leq \Pr[X_t < P(e_t)] = x_e \cdot (1 - s(\varepsilon)) \leq x_e,
\end{equation*}
where the equality $\Pr[X_t < P(e_t)] = x_e \cdot (1 - s(\varepsilon))$ follows by \cref{lemma:analogue_sufficient_cond_for_right_marginal}. 

On the other hand, $ \Pr[\text{$e_t$ matched}]= \Pr[X_t < \hat P(e_t)]$ can be expanded as follows:
\begin{align*}
    \Pr[\text{$e_t$ matched}] &= \Pr[X_t < P(e_t)] - \Pr[X_t < P(e_t) \mid \hat P(e_t) \neq P(e_t)] \cdot \Pr[\hat P(e_t) \neq P(e_t)]. 
\end{align*}
By \cref{lemma:analogue_sufficient_cond_for_right_marginal}, we have that $\Pr[X_t < P(e_t)] = x_e \cdot (1-s(\varepsilon))$. Moreover, using the definition of $P(e_t)$ and the fact that $F_t(u), F_t(v) \geq s(\varepsilon)/4$ (\cref{obs:analogue_Ftv_and_Pet_bounds}), we have that, deterministically:
\begin{equation*}
    P(e_t) = x_e \cdot (1 - s(\varepsilon)) \cdot 
         \frac{1}{F_t(u) \cdot F_t(v)} \leq x_e \cdot \frac{16}{s^2(\varepsilon)}.
\end{equation*}
Consequently, $\Pr[X_t < P(e_t) \mid \hat{P}(e_t) \neq P(e_t)] \leq x_e \cdot \frac{16}{s^2(\varepsilon)}.$
Finally, we leverage our martingale to upper bound the probability of $\hat{P}(e_t)\neq P(e_t)$, as follows. Fix an endpoint $v\in e_t$, and consider the martingale $Y$ associated with $F_t(v)$. First, we note that $Y_0 := \sum_{u_i \in S_t} x_e \cdot (1-s(\varepsilon)) \leq 1-s(\varepsilon)$. 
By \Cref{lemma:analogue_sufficient_cond_for_Phat_equal_P,lemma:analogue_main_ineq}, we want to prove that $Y_m \leq 1 - s(\varepsilon)/3$ is unlikely, and so it is sufficient to upper bound $\Pr[|Y_0 - Y_m| \geq 2 s(\varepsilon)/3]$. By Freedman's inequality, and taking $s(\varepsilon) = C \cdot \sqrt[4]{\varepsilon} \cdot \sqrt[2]{\left(\ln \frac{1}{\varepsilon}\right)}$  (for some sufficiently large constant $C$), we obtain:
\begin{align*}
    \Pr[|Y_0 - Y_m| \geq 2 s(\varepsilon)/3] \leq 
    2 \exp \left(  - \frac{\left(\frac{2}{3}s(\varepsilon) \right)^2}{2 \left( 128 \ln \left( \frac{4}{s(\varepsilon)} \right) \cdot \frac{\varepsilon}{(s(\varepsilon))^2} \right) + \frac{8 \varepsilon}{s(\varepsilon)} \cdot \left(\frac{2}{3}s(\varepsilon) \right)} \right) \leq 2 \eps^5.
\end{align*}
Consequently, by union bounding over both endpoints of $e_t$, we have that:
\begin{align*}
    \Pr[\hat{P}(e_t)\neq P(e_t)]\leq 4\varepsilon^5.
\end{align*}
Combining the above into our expanded form for $\Pr[e_t \textrm{ matched}]$, we obtain the desired inequality:
\begin{equation*}
    \Pr[\text{$e_t$ matched}] \geq x_e \cdot (1-s(\varepsilon)) - x_e \cdot \frac{16}{s^2(\varepsilon)} \cdot 4\varepsilon^5 \geq x_e \cdot (1 - 2s(\varepsilon)),
\end{equation*}
where the last inequality follows because it is equivalent to $C^3 \cdot  \varepsilon^{3/4} \left( \ln \frac{1}{\varepsilon} \right)^{3/2} - 64 \cdot \varepsilon^5 \geq 0$. This can be verified directly for $C \geq 40$ and $0 < \varepsilon \leq 0.99$. 
\end{proof}

%% file: Sections/unknown.tex
\subsection{Application: Online Matching with Unknown \texorpdfstring{$\Delta$}{Delta}}
\label{sec:unknown_delta}

In this section, we point out an immediate application of our generalized fractional matching rounding algorithm (\cref{thm:ronding_matchings_theorem,alg:edge-arrival_extension}). Here we consider the online edge coloring problem, but where the maximum degree $\Delta$ is initially unknown. Again, using $2\Delta-1$ colors is easy by a greedy algorithm which always colors an edge with the smallest available color. In contrast to the case of known $\Delta$, it is proven that a $1.606$-competitive algorithm is the best we can hope for when it comes to unknown $\Delta$, even in the case of large $\Delta = \omega(\log n)$ and vertex arrivals \cite{cohen2019tight}. We approach this lower bound, and in so doing obtain the first online algorithm beating the naive $(2\Delta-1)$-edge-coloring algorithm for unknown~$\Delta$, under general vertex arrivals.

Given our rounding scheme of \Cref{thm:ronding_matchings_theorem}, our result for unknown $\Delta$ follows directly from known reductions introduced in \cite{cohen2019tight} (see also the end of \cite[Chapter 6]{wajc2020matching}). For completeness, we provide a sketch including pointers to the relevant ingredients in \cite{cohen2019tight}.

\begin{theorem} \label{thm:unknown_delta}
    There exists an online algorithm which on $n$-vertex general graphs with maximum degree $\Delta$ only known to satisfy a lower bound $\Delta \geq \Delta'=\omega(\log n)$, with the graph revealed vertex-by-vertex, computes a $(1.777+o(1)) \cdot \Delta$-edge-coloring with high probability.
\end{theorem}

\begin{proof}[Proof (Sketch)]
\cite{cohen2019tight} study the following relaxation of edge coloring; 
in this relaxation a fractional  $\alpha\Delta$-edge-coloring consists of $\alpha\Delta$ many fractional matchings such that each edge is matched to an extent of (at least) one when summed across all fractional matchings. 
In their terminology, a graph is shown to be \emph{fractionally $k$-edge-colorable} if the following LP has a solution.
\begin{align*}
\sum_{c \in [k]} x_{e,c} &= 1  \qquad \qquad \forall e \in E\\
\sum_{e\ni v} x_{e,c} &\leq 1 \qquad \qquad \forall v\in V, c \in [k] \\
x_{e,c} &\geq 0 \qquad \qquad  \forall e\in E, c \in [k]
\end{align*}
For graphs with unknown degree, \cite{cohen2019tight} provide online fractional edge coloring algorithms using $e/(e-1)\Delta$ and $1.777\Delta$ matchings under one-sided vertex arrivals in bipartite graphs and arbitrary vertex arrivals in general graphs, respectively.
Both factional algorithms maintain collections of fractional matchings $\{x_{e,c}\}_c$ with bounded $\ell_{\infty}$ norm, $\max_e x_{e,c} = o(1)$, whenever $\Delta=\omega(1)$. 
\cite{cohen2019tight} further provide (see their Algorithm 2) a rounding framework to round these; they show how to convert online fractional $\alpha\Delta$-edge-coloring algorithms with the above bounded $\ell_{\infty}$ norm guarantee into (randomized) online $(\alpha+o(1))\Delta$-edge-coloring algorithms
for graphs  with unknown maximum degree $\Delta$ satisfying $\Delta\geq \Delta'=\omega(\log n)$, with $\Delta'$ known. (The above $o(1)$ term is of the form $\textrm{poly}(\log n/\Delta')$.)
An important ingredient for their framework is a $(1+o(1))$-approximate rounding scheme for online fractional matchings $\vec{x}$ with $\max_e x_e = o(1)$.
Such a rounding scheme for fractional matchings under one-sided vertex arrivals in bipartite graphs was given by \cite{cohen2018randomized} (see \cite[Chapter 5]{wajc2020matching}). 
Our \Cref{thm:ronding_matchings_theorem} provides such a rounding scheme under edge arrivals (and hence also under vertex arrivals) in general graphs.
Combining our new rounding scheme with the rounding framework of \cite{cohen2019tight} then allows us to round their online fractional $1.777\Delta$-edge-coloring algorithm and obtain an online $(1.777+o(1))\Delta$-edge-coloring algorithm under general vertex arrivals, as claimed.
\end{proof}

%% file: main.bbl
\newcommand{\etalchar}[1]{$^{#1}$}
\begin{thebibliography}{HMRAR98}

\bibitem[AGKM11]{aggarwal2011online}
Gagan Aggarwal, Gagan Goel, Chinmay Karande, and Aranyak Mehta.
\newblock Online vertex-weighted bipartite matching and single-bid budgeted
  allocations.
\newblock In {\em Proceedings of the 22nd Annual ACM-SIAM Symposium on Discrete
  Algorithms (SODA)}, pages 1253--1264, 2011.

\bibitem[AMSZ03]{aggarwal2003switch}
Gagan Aggarwal, Rajeev Motwani, Devavrat Shah, and An~Zhu.
\newblock Switch scheduling via randomized edge coloring.
\newblock In {\em Proceedings of the 44th Symposium on Foundations of Computer
  Science (FOCS)}, pages 502--512, 2003.

\bibitem[BC21]{blanc2021multiway}
Guy Blanc and Moses Charikar.
\newblock Multiway online correlated selection.
\newblock In {\em Proceedings of the 62nd Symposium on Foundations of Computer
  Science (FOCS)}, pages 1277--1284, 2021.

\bibitem[BDBK{\etalchar{+}}94]{ben1994power}
Shai Ben-David, Allan Borodin, Richard Karp, Gabor Tardos, and Avi Wigderson.
\newblock On the power of randomization in on-line algorithms.
\newblock {\em Algorithmica}, 11(1):2--14, 1994.

\bibitem[BDG16]{freedman}
Nikhil Bansal, Daniel Dadush, and Shashwat Garg.
\newblock An algorithm for koml\'{o}s conjecture matching banaszczyk's bound.
\newblock {\em SIAM Journal on Computing}, 48, 2016.

\bibitem[BGW21]{bhattacharya2021online}
Sayan Bhattacharya, Fabrizio Grandoni, and David Wajc.
\newblock Online edge coloring algorithms via the nibble method.
\newblock In {\em Proceedings of the 32nd Annual ACM-SIAM Symposium on Discrete
  Algorithms (SODA)}, pages 2830--2842, 2021.

\bibitem[BMM12]{bahmani2012online}
Bahman Bahmani, Aranyak Mehta, and Rajeev Motwani.
\newblock Online graph edge-coloring in the random-order arrival model.
\newblock {\em Theory of Computing}, 8(1):567--595, 2012.

\bibitem[BNMN92]{bar1992greedy}
Amotz Bar-Noy, Rajeev Motwani, and Joseph Naor.
\newblock The greedy algorithm is optimal for on-line edge coloring.
\newblock {\em Information Processing Letters (IPL)}, 44(5):251--253, 1992.

\bibitem[BSVW24]{blikstad2023simple}
Joakim Blikstad, Ola Svensson, Radu Vintan, and David Wajc.
\newblock Simple and optimal online bipartite edge coloring.
\newblock In {\em Proceedings of the 7th Symposium on Simplicity in Algorithms
  (SOSA)}, 2024.

\bibitem[Chr23]{christiansen2023power}
Aleksander Bj{\o}rn~Grodt Christiansen.
\newblock The power of multi-step vizing chains.
\newblock In {\em Proceedings of the 55th Annual ACM Symposium on Theory of
  Computing (STOC)}, pages 1013--1026, 2023.

\bibitem[CPW19]{cohen2019tight}
Ilan~Reuven Cohen, Binghui Peng, and David Wajc.
\newblock Tight bounds for online edge coloring.
\newblock In {\em Proceedings of the 60th Symposium on Foundations of Computer
  Science (FOCS)}, pages 1--25, 2019.

\bibitem[CW18]{cohen2018randomized}
Ilan~Reuven Cohen and David Wajc.
\newblock Randomized online matching in regular graphs.
\newblock In {\em Proceedings of the 29th Annual ACM-SIAM Symposium on Discrete
  Algorithms (SODA)}, pages 960--979, 2018.

\bibitem[DR96]{dubhashi1996balls}
Devdatt Dubhashi and Desh Ranjan.
\newblock Balls and bins: A study in negative dependence.
\newblock {\em BRICS Report Series}, 3(25), 1996.

\bibitem[Edm65]{edmonds1965maximum}
Jack Edmonds.
\newblock Maximum matching and a polyhedron with 0, 1-vertices.
\newblock {\em Journal of research of the National Bureau of Standards B},
  69(125-130):55--56, 1965.

\bibitem[FHTZ20]{fahrbach2020edge}
Matthew Fahrbach, Zhiyi Huang, Runzhou Tao, and Morteza Zadimoghaddam.
\newblock Edge-weighted online bipartite matching.
\newblock In {\em Proceedings of the 61st Symposium on Foundations of Computer
  Science (FOCS)}, pages 412--423, 2020.

\bibitem[Fre75]{freedman1975tail}
David~A Freedman.
\newblock {On Tail Probabilities for Martingales}.
\newblock {\em The Annals of Probability}, 1975.

\bibitem[GHH{\etalchar{+}}21]{gao2021improved}
Ruiquan Gao, Zhongtian He, Zhiyi Huang, Zipei Nie, Bijun Yuan, and Yan Zhong.
\newblock Improved online correlated selection.
\newblock In {\em Proceedings of the 62nd Symposium on Foundations of Computer
  Science (FOCS)}, pages 1265--1276, 2021.

\bibitem[GKM{\etalchar{+}}19]{gamlath2019online}
Buddhima Gamlath, Michael Kapralov, Andreas Maggiori, Ola Svensson, and David
  Wajc.
\newblock Online matching with general arrivals.
\newblock In {\em Proceedings of the 60th Symposium on Foundations of Computer
  Science (FOCS)}, pages 26--38, 2019.

\bibitem[HKT{\etalchar{+}}20]{huang2020fully}
Zhiyi Huang, Ning Kang, Zhihao~Gavin Tang, Xiaowei Wu, Yuhao Zhang, and Xue
  Zhu.
\newblock Fully online matching.
\newblock {\em Journal of the ACM (JACM)}, 67(3):1--25, 2020.

\bibitem[HMRAR98]{habib1998probabilistic}
Michel Habib, Colin McDiarmid, Jorge Ramirez-Alfonsin, and Bruce Reed.
\newblock {\em Probabilistic methods for algorithmic discrete mathematics},
  volume~16.
\newblock Springer Science \& Business Media, 1998.

\bibitem[HTWZ20]{huang2020fully2}
Zhiyi Huang, Zhihao~Gavin Tang, Xiaowei Wu, and Yuhao Zhang.
\newblock Fully online matching ii: Beating ranking and water-filling.
\newblock In {\em Proceedings of the 61st Symposium on Foundations of Computer
  Science (FOCS)}, page To appear, 2020.

\bibitem[HZZ20]{huang2020adwords}
Zhiyi Huang, Qiankun Zhang, and Yuhao Zhang.
\newblock Adwords in a panorama.
\newblock In {\em Proceedings of the 61st Symposium on Foundations of Computer
  Science (FOCS)}, page To appear, 2020.

\bibitem[Kah96]{kahn1996asymptotically}
Jeff Kahn.
\newblock Asymptotically good list-colorings.
\newblock {\em Journal of Combinatorial Theory, Series A}, 73(1):1--59, 1996.

\bibitem[KLS{\etalchar{+}}22]{kulkarni2022online}
Janardhan Kulkarni, Yang~P Liu, Ashwin Sah, Mehtaab Sawhney, and Jakub
  Tarnawski.
\newblock Online edge coloring via tree recurrences and correlation decay.
\newblock In {\em Proceedings of the 54th Annual ACM Symposium on Theory of
  Computing (STOC)}, pages 2958--2977, 2022.

\bibitem[KVV90]{karp1990optimal}
Richard~M Karp, Umesh~V Vazirani, and Vijay~V Vazirani.
\newblock An optimal algorithm for on-line bipartite matching.
\newblock In {\em Proceedings of the 22nd Annual ACM Symposium on Theory of
  Computing (STOC)}, pages 352--358, 1990.

\bibitem[MSVV07]{mehta2007adwords}
Aranyak Mehta, Amin Saberi, Umesh Vazirani, and Vijay Vazirani.
\newblock Adwords and generalized online matching.
\newblock {\em Journal of the ACM (JACM)}, 54(5):22, 2007.

\bibitem[NSW23]{naor2023online}
Joseph~(Seffi) Naor, Aravind Srinivasan, and David Wajc.
\newblock Online dependent rounding schemes.
\newblock {\em arXiv preprint arXiv:2301.08680}, 2023.

\bibitem[SW21]{saberi2021greedy}
Amin Saberi and David Wajc.
\newblock The greedy algorithm is \emph{not} optimal for on-line edge coloring.
\newblock In {\em Proceedings of the 48th International Colloquium on Automata,
  Languages and Programming (ICALP)}, pages 109:1--109:18, 2021.

\bibitem[Viz64]{vizing1964estimate}
Vadim~G Vizing.
\newblock On an estimate of the chromatic class of a p-graph.
\newblock {\em Diskret analiz}, 3:25--30, 1964.

\bibitem[Viz76]{vizing1976vertex}
Vadim~G Vizing.
\newblock Vertex colorings with given colors.
\newblock {\em Diskret. Analiz}, 29:3--10, 1976.

\bibitem[Waj17]{wajc2017negative}
David Wajc.
\newblock Negative association -- definition, properties, and applications.
\newblock \url{http://www.cs.cmu.edu/\~dwajc/notes/Negative Association.pdf},
  2017.

\bibitem[Waj20]{wajc2020matching}
David Wajc.
\newblock {\em Matching Theory Under Uncertainty}.
\newblock PhD thesis, Carnegie Mellon University, 2020.

\end{thebibliography}
